\documentclass[11pt]{article}

\usepackage{amssymb,amsmath,amsfonts,eurosym,geometry,ulem,graphicx,caption,color,setspace,sectsty,comment,footmisc,caption,natbib,pdflscape,subfigure,array,hyperref}

\usepackage{graphicx}
\usepackage{authblk}
\usepackage{amsfonts}
\usepackage{amsthm}
\usepackage{amsmath}
\usepackage{amscd}
\usepackage[latin2]{inputenc}
\usepackage{t1enc}
\usepackage[mathscr]{eucal}
\usepackage{indentfirst}
\usepackage{graphicx}
\usepackage{graphics}
\usepackage{pict2e}
\usepackage{epic}
\usepackage{url}
\usepackage{comment}
\numberwithin{equation}{section}
\usepackage{natbib} % bibliography package
\usepackage{hyperref}
\hypersetup{
    colorlinks=true,
    linkcolor = red,
    citecolor = blue
}
\usepackage[T1]{fontenc}
\usepackage{mathtools}
\usepackage{blkarray, bigstrut} %

\DeclareMathOperator*{\argmin}{arg\,min}

\theoremstyle{plain}
\newtheorem{Th}{Theorem}[section]
\newtheorem{Lemma}[Th]{Lemma}

\newtheorem{Prop}[Th]{Proposition}

\theoremstyle{definition}
\newtheorem{Def}[Th]{Definition}

\newtheorem{?}[Th]{Problem}
\newtheorem{Ex}[Th]{Example}

\newcolumntype{L}[1]{>{\raggedright\let\newline\\arraybackslash\hspace{0pt}}m{#1}}
\newcolumntype{C}[1]{>{\centering\let\newline\\arraybackslash\hspace{0pt}}m{#1}}
\newcolumntype{R}[1]{>{\raggedleft\let\newline\\arraybackslash\hspace{0pt}}m{#1}}

\begin{document}

\begin{titlepage}
\title{The Probabilistic Serial and Random Priority Mechanisms\\
with Minimum Quotas}
\author[]{Marek Bojko\thanks{Faculty of Economics and Fitzwilliam College, University of Cambridge. Email: marek.bojko@outlook.com.
Most of this paper is based on my undergraduate dissertation written at the University of Glasgow and supervised by Herve Moulin, to whom I am indebted for his valuable guidance and insightful discussions. This work also benefited from the comments from Aytek Erdil, Yehuda John Levy, and Nick Scholz. All errors are mine.
}}
%\affil[]{University of Cambridge}
\date{}
\maketitle
\begin{abstract}
\noindent Consider the problem of assigning indivisible objects to agents with strict ordinal preferences over objects, where each agent is interested in consuming at most one object, and objects have integer minimum and maximum quotas. We define an assignment to be feasible if it satisfies all quotas and assume such an assignment always exists. The Probabilistic Serial (PS) and Random Priority (RP) mechanisms are generalised based on the same intuitive idea: Allow agents to consume their most preferred available object until the total mass of agents yet to be allocated is exactly equal to the remaining amount of unfilled lower quotas; in this case, we restrict agents' menus to objects which are yet to fill their minimum quotas. We show the mechanisms satisfy the same criteria as their classical counterparts: PS is ordinally efficient, envy-free and weakly strategy-proof; RP is strategy-proof, weakly envy-free but not ordinally efficient. \\
\vspace{0in}
\vspace{0in}

\bigskip
\end{abstract}
\setcounter{page}{0}
\thispagestyle{empty}
\end{titlepage}
\pagebreak \newpage

\newpage

\section{Introduction}

This paper considers the problem of assigning indivisible objects to agents with ordinal preferences over the objects, who utilise at most one object, and where no monetary transfers are possible. This class of problems is called \textit{the assignment problem}. In this paper, objects have both maximum and minimum quotas, and we assume such constraints are binding while a feasible solution always exists. Although matching theory has been extensively developed for instances when agents or objects have maximum quotas,\footnote{See e.g. \cite{roth1992two,bogomolnaia2001new,papai2000strategyproof} for two-sided and one-sided matching, respectively.} the literature on matching with minimum quotas is only recent.

We take the Student-Project Allocation problem as our motivating example, in which students are agents with preferences over projects and projects are passive objects. Upper and lower quotas on projects are motivated by load-balancing and oftentimes associated group components.\footnote{For example, undergraduate students of mathematics, computer science and medicine at the University of Glasgow complete their final-year dissertation or semester papers in the form of a research project led by a lecturer. Before the beginning of the relevant semester, supervisors publish a short description of the topic of the project which they are available to supervise. Each project has an associated minimum and maximum number of students supervisors can support. At the end of the preceding academic year, students indicate their preferences by ordering projects which they would like to write their undergraduate dissertation on, based on which they are subsequently allocated a project.} Other motivating examples include cadet assignment to military branches where each branch requests a certain minimum number of cadets \citep{sonmez2013bidding,sonmez2013matching}; assignment of students to schools, labs, and tutorials where load-balancing is considered \citep{fragiadakis2016strategyproof}; assignment of jobs to workers and rooms to housemates when each room requires a minimum number of occupiers (for example, to offset fixed costs associated with purchasing the furniture and heating).

Randomisation is common and preferred in many real-world settings in which fairness is considered. It is often used as a means of breaking ties in many allocation problems. We consider only ordinal lottery mechanisms, in which agents reveal their preferences only over objects instead of over lotteries.\footnote{The terminology of \textit{ordinal lottery mechanisms} was used in \cite{bogomolnaia2001new}. In contrast to ordinal mechanisms, \textit{cardinal lottery mechanisms} elicit Von-Neuman utility functions over lotteries. The Random Priority and Probabilistic Serial mechanisms, as defined below, are ordinal mechanisms. The pseudo-market mechanism of \cite{hylland1979efficient}, which adapts the competitive equilibrium  with equal incomes (CEEI) to the random assignment problem, is probably the most widely-considered cardinal lottery mechanism for the random assignment problem.} This is supported by experimental evidence of limited rationality of agents for whom revealing preferences over lotteries is usually too complex \citep{kagel2016handbook}. Ordinal preferences induce the \textit{first-order stochastic dominance} relation, a partial ordering over deterministic objects, which is used to compare random allocations.

We consider fairness, efficiency, and incentive-compatibility as our design goals. A random allocation mechanism is \textit{strategy-proof} if the random assignment under reporting one's true preferences always stochastically dominates the random assignment under any other action, in other words, it is a weakly dominant strategy to report true preferences; it is \textit{envy-free} if every agent prefers her own random assignment to the random assignment of any other agent. We consider \textit{ordinal efficiency} as our efficiency concept in the context of ordinal preferences. A random assignment is ordinally efficient if it is not stochastically dominated for all agents by any other feasible random assignment \citep{bogomolnaia2001new}.

The Probabilistic Serial (PS) mechanism of \cite{bogomolnaia2001new} and the Random Priority mechanisms have been extensively studied (see e.g. \cite{zhou1990conjecture,abdulkadiroglu1998random,bogomolnaia2001new,budish2013designing} among many others).\footnote{The Random Priority mechanism is often called Random Serial Dictatorship in the literature, see e.g. \cite{abdulkadiroglu1998random}.} The PS mechanism for the classical assignment problem considers each object as divisible, where a fractional assignment means the probability with which an agent receives a particular object.\footnote{The classical assignment problem refers to a problem of assigning $n$ indivisible objects to $n$ agents, where each agent is interested in consuming at most one object. \cite{bogomolnaia2001new} remark that their model and all results for the PS mechanism in their paper hold for an assignment problem with upper quotas and no outside option.} Time runs continuously between 0 and 1, and each agent is allowed to 'eat' from her most preferred available object at the constant unit speed. Once an object is fully consumed, the agents continue eating their next most preferred available object. The random assignment is ordinally efficient, envy-free, and satisfies a weaker notion of incentive-compatibility, namely \textit{weak strategy-proofness}, which says that no agent can obtain a random assignment strictly stochastically dominating the random assignment she would obtain under truthful reporting \citep{bogomolnaia2001new}. The RP mechanism draws an ordering of the agents from the uniform distribution and then lets the first agents select her most preferred object, the second agent her most preferred object from the remaining objects, and so forth. The RP mechanism is strategy-proof, treats equals equally; it is ex-post efficient but may incur unambiguous efficiency loss ex-ante, and so is not ordinally efficient.

We consider a natural generalisation of these two mechanisms to our domain. The Random Priority mechanism under Lower Quotas (RPLQ) draws an ordering of agents randomly uniformly and allows agents to choose sequentially according to this ordering. Throughout the execution, we keep track of the number of copies of objects we still need to allocate to obtain a feasible solution. If the total number of agents yet to be assigned is precisely equal to this number, we restrict their menu to objects with unfilled lower quotas. As a continuous analogue, in the Probabilistic Serial mechanism under Lower Quotas (PSLQ), we restrict the menu of agents to objects with unfilled lower quotas precisely when if the agents continued with the current eating pattern a little longer, we would obtain an unfeasible final allocation. Both mechanisms retain the properties of their classical counterparts.

Weak strategy-proofness of our PSLQ offers an insight into strategic issues in the random assignment problem under additional constraints. \cite{katta2006solution} showed that when agents are allowed to report indifferences in their preferences, there is no mechanism which is ordinally efficient, envy-free, and weakly strategy-proof. In a recent work, \cite{ashlagi2020assignment} consider the assignment of students to schools under distributional constraints, where each school imposes quotas on subsets of students according to their type. They generalise the PS mechanism based on the same underlying principle of appropriately restricting the menu offered to students, and show that there is no ordinally efficient, within-type envy-free and weakly strategy-proof mechanism in their setting. With respect to \cite{ashlagi2020assignment}, we may treat the problem with upper and lower quota on objects as a middle ground between the classical random assignment problem and the random assignment problem with distributional constraints. It is therefore interesting to see that albeit we introduce significant constraints, the positive result of the existence of an ordinally efficient, envy-free, and weakly strategy-proof mechanism is retained. Our proof hinges upon integer quotas, motivated by the real-world settings. We also prove that if quotas are non-negative real numbers,\footnote{We always require that the upper quota of any project is at least as large as the corresponding lower quota.} there is no ordinally efficient, envy-free and weakly strategy-proof random assignment mechanism.

We provide an extension of the model to the random assignment problem with multiple indivisible objects, in which every agent receives up to $q$ objects. The extension of both mechanisms is very straightforward: create $q$ 'clones' of each agent with the same preferences as the original agent. The corresponding generalisation of the PS mechanism fails to be incentive-compatible even in the weak sense because of the impossibility result of \cite{kojima2009random}; however, the rest of the properties are retained.

The rest of the paper is organised as follows. Section~\ref{sec: Related Literature} provides a brief overview of the related literature on matching and assignment under minimum quotas, Section~\ref{sec: Model} introduces the notation and formalises the theoretical framework used in this paper, Sections~\ref{sec: Random Priority under Lower Quotas} and ~\ref{sec: Probabilistic Serial under Lower Quotas} define the RPLQ and PSLQ mechanisms, respectively, and discuss their properties. Finally, we provide an extension of the model to the random assignment problem with multiple indivisible objects in Section~\ref{sec: Random assignment of multiple indivisible objects}, and discuss open problems and conclude in Section~\ref{sec: Extensions and Conclusion}.

\section{Related Literature}\label{sec: Related Literature}

Assignment mechanisms under maximum quotas have been studied extensively. For example, consider the assignment of projects to employees, of students to schools, of rooms to flatmates, of time slots to users of a common resource \citep{shapley1974cores,roth1984evolution, abdulkadirouglu1999house,abdulkadirouglu2003school,abdulkadirouglu2005boston,abdulkadirouglu2005new, bogomolnaia2001new,budish2013designing}. 

The literature on matching with minimum quotas still represents a fairly narrow stream of the rich matching theory, although in the last few years an increasing number of papers have focused on problems with minimum quotas. Minimum quotas have been considered primarily in the school choice problem with upper and lower quotas and distributional constraints (e.g. \cite{biro2010college,kojima2012school,ehlers2014school,hafalir2013effective,kominers2013designing,fragiadakis2016strategyproof}).\footnote{In the school choice problem, students have preferences over schools and schools have priority lists of students.} Since \cite{biro2010college} proved that, in the presence of lower quotas, stable matching might not exist, many of these papers consider soft bounds.\footnote{Flexible limits which are controlled dynamically by schools. For example, schools may opt for a dynamic priority ordering, in which they give high priority to student types who do not fill their floors, student types who fill the lower but not the upper quotas would receive medium priority, and the rest of the students would have the lowest priority. In our setting, projects have no priority orderings, and we only consider hard bounds, i.e. all constraints must be satisfied to define an assignment feasible.} In this paper, we only consider hard constraints. Moreover, the aforementioned papers consider only two-sided matching markets. 

\cite{arulselvan2018matchings} study the maximum-weight many-to-one bipartite matching in the context of the Student-Project Allocation problem, where one side of vertices in the bipartition has upper and minimum quotas, from an algorithmic perspective. They characterise when an instance is NP-hard and when tractable, and define an efficient algorithm for the latter. However, they focus on cardinal utilities and do not consider incentives, fairness, nor efficiency (concerning the welfare of agents). \cite{monte2013matching} consider a related problem, in which assignment is feasible if every object satisfies upper capacity and either it satisfies its lower quota or is assigned to no agent, which is dubbed as \textit{closing projects}. They show the deterministic Priority mechanism, in their paper called Serial Dictatorship, fails to be Pareto efficient and strategy-proof, and describe a mechanism which rectifies it. In this paper, we define an assignment to be feasible if all objects satisfy their quota, and we assume such an assignment always exists. We justify this by noting that in the Student-Project Allocation, it is undesirable to leave students unassigned as it is in most cases compulsory for a student to receive a project - and schools adjust the quotas ex-ante so that a feasible assignment exists. We provide a discussion of a possible extension of our model to the setting with closing projects in Section~\ref{sec: Extensions and Conclusion}.

\cite{budish2013designing} and \cite{fujishige2018random} extend the random assignment problem to domains with bihierarchical distributional and submodular constraints,\footnote{A bihierarchical constraint structure refers to a disjoint union of two laminar families. For more details, see Proporition~\ref{implementable} and the accompanying discussion.} respectively, but consider only upper quotas in their PS mechanisms. \cite{ashlagi2020assignment} extend the Priority and PS mechanisms for the school choice problem where each school imposes upper and lower quotas on subsets of students according to their type. Their generalisation of the PS mechanism is based on the same general idea of restricting the menu of students when constraints become tight. Our contribution over their paper is threefold: (i) our simpler constraint structure allows us to formulate and solve the problem using simple analysis without abstract linear programs and thus allows us to develop core intuition behind the eating algorithm under additional constraints; (ii) we characterize ordinal efficiency in our setting; and (iii) we prove the PSLQ mechanism in our setting is weakly strategy-proof, in contrast to the generalised PS mechanism of \cite{ashlagi2020assignment}.

\section{The model}\label{sec: Model}

We use the language of assigning students to projects. Let $N =[n]=\{1,\dots,n\}$ be the set of students and $P=\{p_1,\dots,p_k\}$ the set of projects.\footnote{With a slight abuse of notation, we may denote by $P$ the set of project seats, unless otherwise stated.} Suppose each project has an upper and a lower quota, $u(p) \in \mathbb{N}$ and $l(p) \in \mathbb{N}_0, u(p)\geq l(p)$, respectively. We assume that $\sum_{p \in P}l(p) \leq n \leq \sum_{p \in P}u(p)$  throughout this paper to ensure that a feasible assignment exists. Let ${l}= [l(p)]_{p \in P}$, ${u}= [u(p)]_{p \in P}$ be vectors of lower and upper quotas, respectively.

Each student $i$ is endowed with strict preferences $\succ_i$ over $P$. We assume that all projects are acceptable to all students. We denote this domain of preferences by $\mathcal{P}$. We denote by $\succ = (\succ_i)_{i \in N}$ a preference profile. For any preference profile $\succ$ and student $i$, we denote by $\succ_{-i}$ the preference profile obtained by disregarding student $i$. The set $\mathcal{P}^n$ denotes the set of all preference profiles for the set of students $N$ over the set of projects $P$.

We denote a market by a tuple $(N,P,{l},{u},\succ)$.

The choice $\phi_i(Q)$ of student $i$ from the subset of projects $Q \subseteq P$ is the best project among $Q$ according to $\succ_i$. That is
$$\phi_i(Q) = p \iff p \in Q \text{ and } p \succ_i p', \forall p' \in Q\setminus \{p\}$$

Given a market $(N,P,{l},{u},\succ)$, an assignment is a correspondence $\mu: N \cup P \to N \cup P$, such that the following holds 
\begin{enumerate}
    \item $\mu(i) \in P, \forall i \in N$
    \item $\mu(p) \in 2^N, \forall p \in P$
    \item $\forall i \in N$ and $p \in P$, we have $\mu(i)=p$ if and only if $i \in \mu(p)$
\end{enumerate}

An assignment $\mu$ is feasible if every project satisfies its quotas, that is if $\forall p \in P, l(p) \leq |\mu(p)| \leq u(p)$. We denote by $\mathcal{D}$ the set of feasible deterministic assignments.

We can describe a feasible deterministic assignment with an $n \times k$ zero-one row-stochastic matrix $X = [x_{ip}]_{i \in N, p \in P}$ which we will call an assignment matrix. We identify rows with students and columns with projects. We have for each entry $x_{ip}$ of $X$ that $x_{ip}=1$ if and only if $i$ is assigned a place at project $p$, and we have
$$l(p) \leq \sum_{i \in N} x_{ip} \leq u(p), \forall p \in P$$
and
$$\sum_{p \in P} x_{ip} = 1, \forall i \in N$$
We denote by $X_i$ the i-th row of $X$, and let $X(t)$ denote an allocation at step/time $t$ of an iterative procedure (and we will use this notation for any matrix in this paper).

A random allocation is a probability distribution over $P$. We will denote the set of random allocations by $\mathcal{L}(P)$. A random assignment is a probability distribution over feasible deterministic assignments; we will denote this set by $\mathcal{L(D)}$. The corresponding convex combination of assignment matrices represents the probability with which a given students is assigned a given project:
$$R = [r_{ip}]_{i \in N, p \in P} \text{, where } R = \sum_{X \in \mathcal{D}} \lambda_X X \text{, such that } \lambda_X \geq 0, \forall X \in \mathcal{D} \text{ and } \sum_{X \in \mathcal{D}}\lambda_X=1$$
The random assignment matrix $R$ is an $n \times k$ row-stochastic matrix; its $i$-th row is the random allocation of student $i$.

\cite{budish2013designing} extend the well-known Birkhoff-Von Neumann theorem for bistochatic matrices, which says that every bistochastic matrix can be written as a convex combination of permutation matrices,\footnote{Zero-one $n \times n$ matrices, in which every row and column sums to one.} to our more general setting of row-stochastic matrices. The following is a direct consequence of their result and says that any random assignment in this paper is implementable as a lottery over deterministic assignments.\footnote{A constraint structure $\mathcal{H}$ is a hierarchy (also known as a laminar family) if $\forall S,T \in \mathcal{H}$ we have that exactly one of the following holds: $S \subset T$ or $T \subset S$ or $S \cap T = \emptyset$. We say a constraint structure $\mathcal{H}$ is a bihierarchy if it can be written as a union of two disjoint hierarchies. \cite{budish2013designing} prove that the condition of having a bihierarchy as a constraint structure is a sufficient condition for implementability of a row sub-stochastic matrix. It is easy to see the constraint structure in our problem is a bihierarchy.}

\begin{Prop}\label{implementable}
Any feasible random assignment matrix can be written as a convex combination of feasible deterministic assignment matrices.
\end{Prop}

Any two probability distribution in $\mathcal{L(D)}$ resulting in the same row-stochastic matrix $R$ will not be distinguished, as they provide the same utility level to every student. Hence, we identify a random assignment with its row-stochastic matrix $R$, its random assignment matrix. We denote by $\mathcal{R}$ the set of feasible random assignments or equivalently random assignment matrices. 

Finally, a deterministic assignment mechanism is a mapping $\Psi: \mathcal{P}^n \to \mathcal{D}$, that is a function that takes any preference profile of students and outputs a feasible deterministic assignment of students to projects. Similarly, a random assignment mechanism is a mapping $\Phi: \mathcal{P}^n \to \mathcal{R}$ that outputs a feasible random assignment. We denote by $\Psi(\succ)$ ($\Phi(\succ)$) the resulting (random) assignment for a preference profile $\succ \in \mathcal{P}^n$. 

\subsection{Efficiency, incentives and fairness}

Preference ordering $\succ_i$ on $P$ induces a partial ordering of the set $\mathcal{L}(P)$ of random allocations that we call the stochastic dominance relation associated with $\succ_i$ and denote by $sd(\succ_i)$. Enumerate $P$ from the most to the least preferred according to $\succ_i$ as $p_1 \succ_i p_2 \succ_i \dots \succ_i p_k$.
We say a vector $R_i \in \mathbb{R}^k$ stochastically dominates $Y_i \in \mathbb{R}^k$ with respect to $\succ_i$ if $\forall t=1,2,\dots,k$, we have
$$\sum_{j=1}^{t} r_{ip_j} \geq \sum_{j=1}^{t} y_{ip_j}$$
Given a preference profile $\succ$, we say that $R \in \mathcal{R}$ stochastically dominates $Y \in \mathcal{R}$ if $R_i \ sd(\succ_i) \ Y_i, \forall i \in N$ and $R \neq Y$.
Random assignment matrix $R$ is ordinarilly efficient if it is not stochastically dominated by any other random assignment matrix.

We turn to fairness. Random assignment $R \in \mathcal{R}$ is envy-free at profile $\succ \in \mathcal{P}^n$ if every student prefers her own allocation to the allocation of any other student, that is $\forall i,j \in N$, we have $R_i \ sd(\succ_i) \ R_j$. It is weakly envy-free if $\forall i,j \in N$, $R_j \ sd(\succ_i) \ R_i \Rightarrow R_i=R_j$. Both of the properties of random assignments introduced above extend to mechanisms.

Random assignment mechanism $\Phi$ is strategy-proof if for any student, reporting her true preferences weakly dominates any other action, that is if for any student $i \in N$, any preference profile $\succ \in \mathcal{P}^n$, and preference ordering $\succ_i' \in \mathcal{P}$, we have $\Phi_i(\succ) \ sd(\succ_i) \ \Phi_i(\succ_i',\succ_{-i})$. It is weakly strategy-proof if, upon misreporting preferences, no student can obtain an allocation strictly dominating the allocation under truthful reporting. Formally, $\Phi$ is weakly strategy-proof if for any student $i \in N$, any preference profile $\succ \in \mathcal{P}^n$, and preference ordering $\succ_i' \in \mathcal{P}$, we have $\Phi_i(\succ_i',\succ_{-i}) \ sd(\succ_i) \ \Phi_i(\succ) \Rightarrow \Phi_i(\succ_i',\succ_{-i})=\Phi_i(\succ)$.

\section{The Random Priority mechanism under Lower Quotas}\label{sec: Random Priority under Lower Quotas}

The following mechanism is a natural extension of the Priority mechanism. A similar version already appeared in \cite{fragiadakis2016strategyproof}. We describe it here for completeness, as it allows us to define the RPLQ mechanism and draw parallels between the extensions of the RP and PS mechanisms. To the best of our knowledge, we are first who consider a uniform lottery over the Priority mechanisms in this setting.

\subsection{The Priority mechanism under Lower Quotas}

Denote by $S_n$ the group of permutations on $n$ letters, which consists of the set of all orderings of students. Fix a market $(N,P,{l},{u},\succ)$ and a permutation $\sigma \in S_n$. The Priority mechanism under Lower Quotas (PrioLQ), which we will denote by $\mu^{\sigma}=PrioLQ(\succ)^{\sigma}$ then proceeds as follows:

\bigskip

\textbf{Step 1:} Student $\sigma(1)$ is assigned a seat at her most preferred project according to $\succ_{\sigma(1)}$, i.e. 
$$\mu^{\sigma}(\sigma(1))=\phi_{\sigma(1)}(P)$$

\hspace{5mm} \vdots\\

\textbf{Step k:} Denote by $P_l^k \subseteq P$ the set of projects which have not satisfied their lower quotas at the beginning of step $k$. If
\begin{equation} \label{cond}
    \sum_{p \in P_l^k} \Bigg( l(p)-\sum_{i \in N}x_{ip}(k) \Bigg) < n-k+1
\end{equation}
student $\sigma(k)$ is assigned a place at her most preferred available project according to her preference ordering $\succ_{\sigma(k)}$. Otherwise, student $\sigma(k)$ is assigned a place at her most preferred project in $P_l^k$. Formally,\footnote{With a slight abuse of notation, we denote by $P\setminus \bigcup_{j=1}^{k-1}\mu^{\sigma}(\sigma(j))$ the set of remaining projects after students $\sigma(1)$ through $\sigma(k-1)$ made their choices.}
$$\mu^{\sigma}(\sigma(k))=
\begin{cases}
    \phi_{\sigma(k)}\left(P\setminus \bigcup_{j=1}^{k-1}\mu^{\sigma}(\sigma(j))\right )   & \text{, condition } \ref{cond} \text{ holds} \\
    \phi_{\sigma(k)}\left(P_l^k\right) &\text{, otherwise}
\end{cases}$$

\bigskip

Since the sets $N$ and $P$ are finite, it is clear the algorithm terminates in finite time. It is also straightforward to check that the resulting assignment is feasible. Since $|S_n|=n!$, we have $n!$ priority mechanisms, each of which is induced by a different permutation of the set of students. Note the ordering selecting a particular allocation need not be unique.

It is apparent that any Priority mechanism does not treat students symmetrically: Student $\sigma(1)$ always gets her most preferred project while $\sigma(n)$ gets whatever is left. One way to restore fairness is to randomise over deterministic assignments, to which we turn in the next section. It could also be resolved using a Master List (ML), an exogenous ordering of the students independent of their preferences. Such ordering might be created, for example, based on cumulative GPAs of students, their attendance record or extracurricular activities. We introduce the following notion of fairness in the deterministic case.\footnote{This notion of fairness is equivalent to justified envy-freeness in two-sided matching markets where one side of the market has priority orderings over the other, e.g. in school choice.} The definition of efficiency reflects the constraints imposed on projects.

\begin{Def}
Assignment $\mu \in \mathcal{D}$ is 
\begin{enumerate}
    \item ML-fair if whenever $\mu(j) \succ_i \mu(i)$ for some $i,j \in N$, $j$ precedes $i$ on the master list.
    \item Minimum-quota-constrained efficient if it is not Pareto-dominated by any other $\mu' \in \mathcal{D}$, where we say $\mu \in \mathcal{D}$ \textit{Pareto dominates} $\mu' \in \mathcal{D}$ at a profile $\succ \in \mathcal{P}^n$ if $\exists i \in N$ such that $\mu(i) \succ_i \mu'(i)$ and $\forall j \in N, \mu(j) \succeq_j \mu'(j)$, where $\mu(j) \succeq_j \mu'(j)$ if $\mu(j) \succ_j \mu'(j)$ or $\mu(j) = \mu'(j)$.
\end{enumerate}

Both of these properties extend to mechanisms. Mechanism $\Psi$ is strongly group strategy-proof if no coalition of students can make all members at least as well off and at least one member strictly better off by jointly misreporting, compared to when all members report truthfully. Formally, $\Psi$ is strongly group strategy-proof if there is no $\succ \in \mathcal{P}^n, S \subseteq N$, and $\succ_S' \in \mathcal{P}^{|S|}$ such that $\Psi_i(\succ_S',\succ_{-S}) \succeq_i \Psi_i(\succ), \forall i \in S$ and $\exists j \in S \text{ such that } \Psi_j(\succ_S',\succ_{-S}) \succ_j \Psi_j(\succ)$.
\end{Def}

\begin{Prop}\label{priolq ML fairness}
For any market $(N,P,{l},{u},\succ)$, PrioLQ is:
\begin{enumerate}
    \item Minimum-quota-constrained efficient
    \item Strongly group strategy-proof
    \item ML-fair
\end{enumerate}
\end{Prop}

We refer the reader to \cite{fragiadakis2016strategyproof} for proofs of (2.) and (3.). The proof of (1.) follows a similar argument as for the classical Priority mechanism and is omitted.

\subsection{The lottery mechanism}

The mechanism proceeds as follows: Fix a preference profile $\succ \in \mathcal{P}^n$, draw at random a permutation of students $\sigma$ from the uniform distribution over $S_n$ and then run $PrioLQ(\succ)^{\sigma}$.

The Random Priority mechanism under Lower Quotas (RPLQ) for the preference profile $\succ$, which we will denote by $RPLQ(\succ)$, is then defined as

\begin{equation*}
RPLQ(\succ)=\frac{1}{n!}\sum_{\sigma \in S_n}PrioLQ(\succ)^{\sigma}    
\end{equation*}

Since RPLQ is a convex combination of feasible Pareto optimal deterministic assignments, it follows that the resulting random assignment is also feasible and ex-post efficient.

\begin{comment}
\begin{Lemma}\label{PrioLQ feasible}
Fix $(N,P,{l},{u},\succ)$. $R=RPLQ(\succ)$ is feasible, that is $\forall p \in P, l(p) \leq \sum_{i \in N}r_{ip} \leq u(p)$.
\end{Lemma}
\begin{proof}
Appendix~\ref{proof PrioLQ feasible}.
\end{proof}
\end{comment}

\begin{Ex}\label{ex PRLQ envy}
Lower quotas generate envy. Suppose $N=[4]$ and $P=\{a,b,c\}$ with no quotas. Let students have preferences $\succ$ as below; call this market $\Gamma$. Then $RPLQ(\Gamma)$ outputs matrix $R$. This random assignment is envy-free, as every student is assigned her top choice. Now suppose $l(b)=2$ and $l(c)=1$ while $l(a)=0$; let us call this market $\zeta$. The resulting random assignment $RPLQ(\zeta)$ is $R'$.

\[
\succ \ = 
    \begin{tabular}{ccc}
    \hline
    1 & 2 & 3,4 \\
    \hline
    a & a & b \\
    b & c & a \\
    c & b & c \\
    \hline
    \end{tabular}
  \qquad
  R=
  \begin{blockarray}{*{3}{c} l}
    \begin{block}{*{3}{>{\footnotesize}c<{}} l}
      a & b & c \\
    \end{block}
    \begin{block}{[*{3}{c}]>{\footnotesize}l<{}}
      1 & 0 & 0 & 1,2 \\
      0 & 1 & 0 & 3,4 \\
    \end{block}
  \end{blockarray}
  \qquad
  R'=
  \begin{blockarray}{*{3}{c} l}
    \begin{block}{*{3}{>{\footnotesize}c<{}} l}
      a & b & c \\
    \end{block}
    \begin{block}{[*{3}{c}]>{\footnotesize}l<{}}
      1/2 & 1/4 & 1/4 & 1 \\
      1/2 & 0 & 1/2 & 2\\
      0 & 7/8 & 1/8 & 3,4\\
    \end{block}
  \end{blockarray}
\]

Observe that $R_1'$ does not stochastically dominate $R_3'$ with respect to $\succ_1$, so this random assignment is not envy-free.
\end{Ex}

RPLQ satisfies a weaker notion of envy-freeness. Strategy-proofness is a natural extension of strategy-proofness of PrioLQ. These results are consistent with the Random Priority mechanism \citep{bogomolnaia2001new}.

\begin{Prop} \label{RPLQ}
For any market $(N,P,{l},{u},\succ)$, the RPLQ mechanism is 
\begin{enumerate}
    \item Strategy-proof
    \item Weakly envy-free
\end{enumerate} 
\end{Prop}

\begin{proof}
Appendix~\ref{proof RPLQ}.
\end{proof}

RPLQ might be ordinarily inefficient; the inefficiency might also be caused solely by lower quotas.

\begin{Ex}\label{ex rplq ordinally inefficient}
Let $N=[6]$ and $P=\{a,b,c,d\}$. Suppose $l(b)=l(c)=2$, and there are no other quotas. Students have preferences $\succ$. The resulting random assignment matrix is $R$ which is stochastically dominated by $R'$.

\[
\succ \ = 
    \begin{tabular}{cc}
    \hline
    1,2,3 & 4,5,6 \\
    \hline
    a & b \\
    b & a \\
    c & d \\
    d & c \\
    \hline
    \end{tabular}
  \qquad
  R=
  \begin{blockarray}{*{4}{c} l}
    \begin{block}{*{4}{>{\footnotesize}c<{}} l}
      a & b & c & d\\
    \end{block}
    \begin{block}{[*{4}{c}]>{\footnotesize}l<{}}
      3/5 & 1/15 & 1/3 & 0 & 1,2,3\\
      0 & 2/3 & 1/3 & 0 & 4,5,6\\
    \end{block}
  \end{blockarray}
  \qquad
  R'=
  \begin{blockarray}{*{4}{c} l}
    \begin{block}{*{4}{>{\footnotesize}c<{}} l}
      a & b & c & d\\
    \end{block}
    \begin{block}{[*{4}{c}]>{\footnotesize}l<{}}
      2/3 & 0 & 1/3 & 0 & 1,2,3\\
      0 & 2/3 & 1/3 & 0 & 4,5,6\\
    \end{block}
  \end{blockarray}
\]
\end{Ex}

\section{The Probabilistic Serial mechanism under Lower Quotas}\label{sec: Probabilistic Serial under Lower Quotas}

The PS mechanism of \cite{bogomolnaia2001new} is a central element in the set of ordinally efficient random assignment mechanisms for the classical random assignment problem. In the eating algorithm, time runs continuously between 0 and 1, and students eat from their most preferred available project with constant unit eating speed. They are required to move from the project they are currently eating if and only if the project is fully eaten, in which case they move to their next most preferred available project. We follow the same basic idea as with PrioLQ: Allow students to eat from their most preferred available project as long as there is enough time left to satisfy all lower quotas. When constraints start to 'bite', we restrict their menu to projects which are yet to satisfy their lower quota. We note the mechanism can be regarded as solving a special case of the problem with distributional constraints as in \cite{ashlagi2020assignment} where all students are of the same type.

\subsection{The eating algorithm}

Similarly to PrioLQ, in the design of the Probabilistic Serial under Lower Quotas (PSLQ), we need to identify the menu presented to a student at any point of the execution and an optimal time when students should be moved to consume a project that is at risk of not satisfying its lower quota. Intuitively, a student should be allowed to consume her most preferred available project as long as all other projects can satisfy their quotas if we continue with the same eating pattern a little longer. 

\begin{Ex}\label{ex pslq 1}
Suppose $N=[5]$ and $P=\{a,b,c\}$. Suppose the projects and students have the following quotas and preferences, respectively,

\[
\begin{tabular}{ c|ccc }
& a & b & c \\
  \hline
$u(\cdot)$ & 2 & 2 & 2\\
$l(\cdot)$ & 1 & 1 & 2 \\
\end{tabular}
\qquad
\succ \ = 
    \begin{tabular}{ccc}
    \hline
    1,2 & 3,4 & 5 \\
    \hline
    a & b & c\\
    b & a & a\\
    c & c & b\\
    \hline
    \end{tabular}
\qquad
R=
  \begin{blockarray}{*{3}{c} l}
    \begin{block}{*{3}{>{\footnotesize}c<{}} l}
      a & b & c \\
    \end{block}
    \begin{block}{[*{3}{c}]>{\footnotesize}l<{}}
      3/4 & 0 & 1/4 & 1,2\\
      0 & 3/4 & 1/4 & 3,4\\
      0 & 0 & 1 & 5\\
    \end{block}
  \end{blockarray}
\]

Observe that student 5 is not able to satisfy the lower quota of $c$ on her own. Therefore, we seek the critical time $t_c \in [0,1)$ such that if we continue with the same eating pattern a little longer, we would not be able to satisfy the lower quota of $c$. As a continuous analogue to condition~\ref{cond} when considering PrioLQ, we posit $t_c$ must satisfy
$$n(1-t_c)=l(c)-t_c$$
that is the remaining mass of students is exactly enough to satisfy the lower quota of $c$ after substracting what student 5 has eaten until time $t_c$. Rearranging, we obtain
$$t_c=\frac{n-l(c)}{n-1}=\frac{3}{4}$$
The resulting random assignment matrix is R, which is clearly feasible.
\end{Ex}

Consider any point $t \in [0,1]$ during the execution of the algorithm and denote by $R(t)$ the assignment matrix at this time (i.e. the entry $r_{ip}(t)$ means how much has student $i$ eaten from project $p$ until time $t$). We must ensure that $R(1)$ is feasible.

In what follows, we define $(\cdot)_+: \mathbb{R} \to \mathbb{R}_{\geq 0}$ to be $(x)_+=\max\{0,x\}$.

Fix $t \in [0,1]$. At $t$, $n(1-t)$ is the remaining mass of students that can be distributed to projects $P$. We will denote such distribution by a tuple $[z_p(t)]_{p \in P} \in \mathbb{R}_{\geq 0}^k$, where $\sum_{p \in P}z_p(t)=n(1-t)$.
For a solution to be feasible, we require
$$z_p(t)+\sum_{i \in N}r_{ip}(t) \geq l(p), \forall p \in P$$
Rearranging, we have
$$z_p(t) \geq l(p)-\sum_{i \in N}r_{ip}(t), \forall p \in P$$
and since $z_p(t) \geq 0, \forall p \in P$ by definition, we have
$$z_p(t) \geq \left( l(p)-\sum_{i \in N}r_{ip}(t) \right)_+ $$
Hence, summing over all projects,

\begin{equation}\label{eq pslq feasibility}
   n(1-t)=\sum_{p \in P}z_p(t) \geq \sum_{p \in P} \left( l(p)-\sum_{i \in N}r_{ip}(t) \right)_+ 
\end{equation}

Note the set of distributions $[z_p(t)]_{p \in P}$, together with the partial assignment $R(t)$, determine the set of feasible random assignment matrices we can obtain at $t=1$.

This motivates the following definition.

\begin{Def}\label{def active project}
Fix $t \in [0,1]$. We say project $p \in P$ is active at $t$ if
(i) $l(p) > \sum_{i \in N}r_{ip}(t)$; or (ii) $l(p) \leq \sum_{i \in N}r_{ip}(t) < u(p)$ and
$$n(1-t) > \sum_{p \in P} \left( l(p)-\sum_{i \in N}r_{ip}(t) \right)_+$$
We denote by $\mathcal{A}(t)$ the set of active projects at time $t$.
\end{Def}

To ensure $R(1)$ is feasible, we limit the choice of students to projects in $\mathcal{A}(t)$ at all $t \in [0,1]$.

\begin{Ex}
Example~\ref{ex pslq 1} continued. As derived in Example~\ref{ex pslq 1}, we have, at $t_c=3/4$,
$$n(1-t_c)=l(c)-t_c$$
Since $\sum_{i \in N}r_{ip}(t) > l(p)$, for $p=a,b$ and $\sum_{i \in N}r_{ic}(t) < l(c)$, the only active project at $t_c$ is $c$.
\end{Ex}

Before we formally define the iterative procedure, let us introduce one more notion. We define $\chi_{ip}$ to be the indicator function which tells us whether student $i$ is eating project $p$ at time $t$.
$$\chi_{ip}(t)=
\begin{cases}
1 & \text{ , } p = \phi_i(\mathcal{A}(t))\\
0 & \text{ , otherwise}
\end{cases}$$

We assume all agents have unit eating speed. Given a preference profile $\succ$, the eating algorithm is defined as the following sequence of recursive steps. Let $t^0=0$, $\mathcal{A}^0 = \mathcal{A}(0)=P$ and $R^0=R(0)=[0]_{i \in N, p \in P}$. Given $t^0,\mathcal{A}^0, R^0,\dots, t^{v-1}, \mathcal{A}^{v-1}=\mathcal{A}(t^{v-1}), R^{v-1}=R(t^{v-1})$, for all $p \in \mathcal{A}^{v-1}$  define:
\bigskip
\begin{enumerate}
    \setlength\itemsep{2em}
    \item $$t^v(p) = \sup\{t \in [0,1]: \sum_{j \in N}\left( r_{jp}^{v-1}+\chi_{jp}(t^{v-1}) (t-t^{v-1}) \right) < u(p) \}$$
    \item $$\tau^v = \min_{p \in \mathcal{A}^{v-1}} t^v(p)$$
    \item $$\lambda = \sup\{t \in [0,1]: n(1-t) > \sum_{p \in P} \left( l(p)-\sum_{i \in N}r_{ip}(t^{v-1}) - \chi_{ip}(t^{v-1})(t-t^{v-1}) \right)_+\}$$
    \item $$t^v = \min\{\tau^v,\lambda\}$$ 
    \item $$r_{ip}^v = r_{ip}^{v-1} + \chi_{ip}(t^{v-1})(t^v-t^{v-1}), \forall i \in N, p \in P$$
    \item If $\lambda > \tau^v$, let 
    $$\mathcal{A}^v = \mathcal{A}^{v-1} \setminus \{p \in \mathcal{A}^{v-1}: t^v(p) = t^v\}$$ \\
    else let 
    $$\mathcal{A}^v = \mathcal{A}^{v-1} \setminus \{p \in \mathcal{A}^{v-1}:\left( l(p)-\sum_{i \in N}r_{ip}(t^{v-1}) -  \chi_{ip}(t^{v-1}) \lambda \right)_+=0\}$$
\end{enumerate}
\bigskip
By construction, $\mathcal{A}^v \subset \mathcal{A}^{v-1}$ for all periods $v$ and so $\exists j \in \mathbb{N}$ such that $\mathcal{A}^j=\emptyset$. We define $PSLQ(\succ)=R^j=R(1)$ for the preference profile $\succ$. Note that the eating algorithm is \textit{anonymous}: the mapping $\succ \to PSLQ(\succ)$ is symmetric from the n preferences $\succ_i$ to the n assignments $R_i$.

Notice that if we define

\begin{equation}\label{eq: once shifted the rest has lower quota 1}
t_c = \inf\{t \in [0,1]: n(1-t) = \sum_{p \in P} \left( l(p)-\sum_{i \in N}r_{ip}(t) \right)_+\}    
\end{equation}
then
\begin{equation}\label{eq: once shifted the rest has lower quota 2}
n(1-t) = \sum_{p \in P} \left( l(p)-\sum_{i \in N}r_{ip}(t) \right)_+, \forall t > t_c    
\end{equation}

Hence, if $t_c<1$, all projects that students cease to eat from at some $t \in (t_c,1]$ have the total assignment equal to their lower quota. As a consequence of this fact, if there is a student shifted to a project $p$ at any $t \in (0,1)$, in particular, it must hold that $t \geq t_c$, and so the total fractional assignment to $p$ is $l(p)$. If $t_c = 1$, we recover the PS mechanism.

\subsection{Ordinal efficiency and envy-freeness}

\cite{bogomolnaia2001new} characterise ordinal efficiency based on acyclicity of the following binary relation defined on the set of projects. 

\begin{Def}
Given a random assignment matrix $R \in \mathcal{R}$ and a preference profile $\succ \in \mathcal{P}$, we define the binary relation $\tau(R,\succ)$ over the set of projects $P$ as follows
$$\forall p,q \in P: p \ \tau(R,\succ) \ q \iff \exists i \in N: p \succ_i q \text{ and } r_{iq}>0$$
The relation $\tau(R,\succ)$ is cyclic if there is a cycle of relations of the form $p_1 \ \tau \ p_2 \ \tau \ ... \ \tau \ p_n  \ \tau \ p_1$ (writing $\tau$ shorthand for $\tau(R,\succ)$).
\end{Def}

In our setting, acyclicity of $\tau$ is not a sufficient criterion for ordinal efficiency, as can be deduced from Example~\ref{ex PRLQ envy}. \cite{kojima2010incentives} consider the problem with maximum quotas and an outside option; ordinal efficiency is then characterised by acyclicity of $\tau(R\succ)$ and non-wastefulness.

\begin{Def}\label{def: wasteful}
Given a preference profile $\succ$, we say random assignment $R \in \mathcal{R}$ is wasteful if $\exists j \in N$ and $p,q \in P$ such that all of the following are satisfied (i) $p \succ_j q$, (ii) $r_{iq}>0$, (iii) $\sum_{i \in N}r_{ip} < u(p)$, and (iv) $\sum_{i \in N}r_{iq} > l(q)$.
\end{Def}

Minimum quotas restrict the set of ordinally efficient random assignments even further, and we require a stronger condition than non-wastefulness.

\begin{Ex}\label{ex: acyclic and wasteful not sufficient}
Let $N=\{1,2\}$, $P=\{a,b,c\}$ and $l(b)=u(b)=1$, while there are no other constraints. Consider the random assignment $R$, which is ordinally inefficient, as it is stochastically dominated by $R'$. Observe that $a \ \tau \ b$ and $b \ \tau \ a$; $\tau$ is acyclic. Moreover, $R$ is not wasteful since $l(b)=r_{1b}+r_{2b}=u(b)$.
\[\succ \ = 
    \begin{tabular}{cc}
    \hline
    1 & 2 \\
    \hline
    a & b \\
    b & c \\
    c & a \\
    \hline
    \end{tabular}
  \qquad
R=
  \begin{blockarray}{*{3}{c} l}
    \begin{block}{*{3}{>{\footnotesize}c<{}} l}
      a & b & c \\
    \end{block}
    \begin{block}{[*{3}{c}]>{\footnotesize}l<{}}
      1/2 & 1/2 & 0 & 1\\
      0 & 1/2 & 1/2 & 2\\
    \end{block}
  \end{blockarray}
 \qquad
 R'=
  \begin{blockarray}{*{3}{c} l}
    \begin{block}{*{3}{>{\footnotesize}c<{}} l}
      a & b & c \\
    \end{block}
    \begin{block}{[*{3}{c}]>{\footnotesize}l<{}}
      1 & 0 & 0 & 1\\
      0 & 1 & 0 & 2\\
    \end{block}
  \end{blockarray}
\]
\end{Ex}

The example motivates the following definition.

\begin{Def}\label{def: wasteful chain}
Given a preference profile $\succ$, random assignment $R \in \mathcal{R}$ contains a wasteful chain 
\begin{equation*}
    \varnothing \rightarrow p_1 \rightarrow i_1 \rightarrow p_2 \rightarrow \dots \rightarrow i_l \rightarrow p_{l+1} \rightarrow \varnothing
\end{equation*}
if there is a sequence of students $i_1,\dots,i_l \in N$ and a sequence of projects $p_1,\dots,p_{l+1} \in P$ such that all of the following are satisfied: (i) $p_j \succ_{i_j} p_{j+1}$, (ii) $r_{i_jp_{j+1}} > 0, \forall 1 \leq j \leq l$; (iii) $\sum_{i \in N} r_{ip_1} < u(p_1)$ and (iv) $\sum_{i \in N} r_{ip_{l+1}} > l(p_{l+1})$.\footnote{Note that if $R$ is wasteful at $\succ$ by student $i$ and projects $p,q$ as in Definition~\ref{def: wasteful}, in particular $i, p$, and $q$ form a wasteful chain. Therefore, $R$ is wasteful at $\succ$ $\Rightarrow$ $R$ contains a wasteful chain at $\succ$.}
\end{Def}

\begin{Ex}
We illustrate Definition~\ref{def: wasteful chain} on Example~\ref{ex: acyclic and wasteful not sufficient}. Observe the sequence $1,2$ of students and $a,b,c$ of projects is a wasteful chain.
\begin{equation*}
    \varnothing \rightarrow a \rightarrow 1 \rightarrow b \rightarrow 2 \rightarrow c \rightarrow \varnothing
\end{equation*}
With respect to feasibility, we can costlessly allow $1$ to consume more of project $a$ as student $2$ profits from substituting for 1 in the consumption of $b$, while project $c$ satisfies its lower quota in either case.

\end{Ex}

We now characterise the set of ordinally efficient random assignments on the strict preference domain.

\begin{Lemma}\label{lemma: characterisation ordinal efficiency}
Random assignment $R \in \mathcal{R}$ is ordinally efficient at $\succ \in \mathcal{P}^n$ if and only if the binary relation $\tau(R,\succ)$ is acyclic and $R$ does not contain a wasteful chain under $\succ$.
\end{Lemma}

\begin{proof}
Appendix~\ref{proof lemma: characterisation ordinal efficiency}.
\end{proof}

We are ready to prove that PSLQ satisfies the same properties of efficiency and fairness as the PS mechanism.

\begin{Prop}\label{PSLQ o-efficient}
For any market $(N,P,{l},{u},\succ)$, PSLQ is 
\begin{enumerate}
    \item Ordinally efficient
    \item Envy-free
\end{enumerate}
\end{Prop}

\begin{proof}
Appendix~\ref{proof PSLQ o-efficient}.
\end{proof}

\subsection{(Weak) Strategy-proofness}

We turn to incentive-compatibility. First, we use Example~\ref{ex PRLQ envy} to illustrate that introducing lower quotas might create incentives for students to misreport their true preferences.

\begin{Ex}
First, note that PSLQ and RPLQ produce exactly the same assignment for market $\Gamma$, which is clearly non-manipulable. Consider the market $\zeta$; $PSLQ(\zeta)$ outputs matrix $R$. However, if student 3 misreports $\succ_3': a \succ_3' b \succ_3' c$, the resulting random assignment is $R'$. Then, $R_3$ does not stochastically dominate $R_3'$ with respect to $\succ_3$.

\[
R=
  \begin{blockarray}{*{3}{c} l}
    \begin{block}{*{3}{>{\footnotesize}c<{}} l}
      a & b & c \\
    \end{block}
    \begin{block}{[*{3}{c}]>{\footnotesize}l<{}}
      1/2 & 1/3 & 1/6 & 1\\
      1/2 & 0 & 1/2 & 2\\
      0 & 5/6 & 1/6 & 3,4\\
    \end{block}
  \end{blockarray}
 \qquad
 R'=
  \begin{blockarray}{*{3}{c} l}
    \begin{block}{*{3}{>{\footnotesize}c<{}} l}
      a & b & c \\
    \end{block}
    \begin{block}{[*{3}{c}]>{\footnotesize}l<{}}
      1/3 & 5/9 & 1/9 & 1\\
      1/3 & 0 & 2/3 & 2\\
      1/3 & 5/9 & 1/9 & 3\\
      0 & 8/9 & 1/9 & 4\\
    \end{block}
  \end{blockarray}
\]

\end{Ex}

The PS mechanism satisfies a weaker notion of incentive-compatibility - weak strategy-proofness. We show PSLQ coincides with the PS mechanism in this avenue as well. Our positive result is in contrast to the impossibility result of \cite{ashlagi2020assignment}, who show that when students are partitioned into groups and each group has quotas on its own, it is impossible to design a mechanism which would satisfy ordinal efficiency, group envy-freeness and weak strategy-proofness.

\begin{Th}\label{PSLQ weakly strategy-proof}
For any market $(N,P,{l},{u},\succ)$, PSLQ is weakly strategy-proof.
\end{Th}

\begin{proof}
Appendix~\ref{proof PSLQ weakly strategy-proof}.
\end{proof}

What might seem surprising at first glance is the reliance of the result on the condition ${l} \in \mathbb{Z}_{\geq 0}^k$ and ${u} \in \mathbb{N}^k$. In the proof in Appendix~\ref{proof PSLQ weakly strategy-proof}, we fix student $i \in N$ and suppose $i$ misreports to $\succ_i'$ such that $PSLQ_i(\succ_i',\succ_{-i}) \ sd(\succ_i) \ PSLQ_i(\succ)$. In order to show that PSLQ is weakly strategy-proof, we need to show $PSLQ_i(\succ_i',\succ_{-i}) = PSLQ_i(\succ)$. Abstracting from integer quotas, we show that this is indeed the case unless under $\succ$, student $i$ is the only one ever eating from a project $p$ such that $i$ is shifted from $p$, and the total fractional allocation of any other project is either its upper or lower quota. If we allow non-integer quotas, we illustrate this class of instances might indeed admit strategic manipulation with an example in Appendix~\ref{sec: Strategic manipulation under non-integer quotas}.\footnote{Note that, assuming we allow non-integer quotas, we do not characterise under what conditions the PSLQ mechanism admits strategic manipulation. We only show the configuration of the parameters described above is a necessary condition.} However, the configuration of parameters described above is incompatible with integer quotas: We allocate total mass of $n \in \mathbb{N}$ students, and the total mass allocated to projects other than $p$ must be an integer, as it is a sum of integer allocations. Finally, notice that the total allocation to $p$ is between 0 and 1 since student $i$ is the only one ever eating $p$ and $i$ is shifted from $p$. Under integer quotas, we therefore conclude $PSLQ_i(\succ_i',\succ_{-i}) = PSLQ_i(\succ)$, and the result follows.

So far, we have shown that PSLQ is weakly strategy-proof if all quotas are integers and it fails to be incentive-compatible otherwise. It turns out it is impossible to construct an ordinally efficient, envy-free and weakly strategy-proof random assignment mechanism if we allow non-integer quotas and consider at least two students and three projects.\footnote{While the result is not directly applicable to the real-world settings described in the Introduction, it shows the tightness of weak strategy-proofness of the PSLQ mechanism under integer quotas: Under small perturbation of the quotas, the eating algorithm is not weakly strategy-proof and we obtain the impossibility result.}

\begin{Prop}\label{impossibility result}
There is no random assignment mechanism satisfying ordinal efficiency, envy-freeness, and weak strategy-proofness for all markets $(N,P,{l},{u},\succ)$ with $n \geq 2$, $k \geq 3$, $l \in \mathbb{R}^k_{\geq 0}$, and $u \in \mathbb{R}^k_{>0}$.
\end{Prop}

\begin{proof}
Appendix~\ref{proof impossibility result}.
\end{proof}

%%%%%%%%%%%%%%%%%%%%%%%%%%%%%%%%%%%%%%%%%%%%%%%%%%%%%%%%%%%%%%
\section{Random assignment of multiple indivisible objects under lower quotas}\label{sec: Random assignment of multiple indivisible objects}

Our model can be extended to solve the random assignment of multiple objects, in which each agent receives more than one object and preferences of agents are additively separable across copies of objects.\footnote{Preferences over subsets of the set of objects $P$ are additively separable over objects if for each agent $i$ there is a cardinal utility function $u_i: P \rightarrow \mathbb{R}$ such that for every subset $Q \subseteq P$, $u_i(Q)=\sum_{p \in Q}u_i(p)$.} Consider, for example, the problem of distributing patients to physicians. Suppose each patient has $q$ visits covered from the health insurance every year, each physician has a minimum and a maximum number of visits they can deliver per year, and suppose each patient prefers to claim as many slots as possible at their preferred physician. 

More formally, we consider the same set-up as in Section~\ref{sec: Model}, only that each agent has an upper bound $q \in \mathbb{N}$ on the number of objects she can be assigned. Let us assume $\sum_{p \in P}l(p) \leq q \cdot n \leq \sum_{p \in P}u(p)$ to ensure feasibility. Each row $i$ of an assignment matrix must sum to $q$. 

The extension of PSLQ to this setting is straightforward: let the epoch be $[0,q]$ instead of $[0,1]$, and hence each agent eats $q$ copies of objects. Equivalently, we could define the mechanism by creating $q$ 'clones' of each agent, with the same preferences as the original agent. \cite{kojima2009random} considers a natural extension of the PS mechanisms to this setting where there is one unit of each good available and no lower quotas, and prove that this mechanism satisfies ordinal efficiency and no envy. It can be shown this also holds when we consider lower quotas. Since \cite{kojima2009random} proved that no mechanism in this setting is ordinally efficient, envy-free and weakly strategy-proof, it follows that this extension of the PSLQ mechanism fails to be incentive-compatible. Similarly, we can extend RPLQ by creating $q$ 'clones' of each student, with the same preferences as the original student. This mechanism violates ordinal efficiency; however, it is strategy-proof and weakly envy-free.

\section{Conclusion}\label{sec: Extensions and Conclusion}

We studied the assignment problem where objects have upper and lower quotas and agents reveal strict ordinal preferences over all objects. The assignment problem with minimum quotas is complex and we considered the problem of distributing agents to objects with the objective of satisfying quotas of all objects, while assuming a feasible assignment always exists. 

Our main contribution is generalising the Random Priority mechanism and the Probabilistic Serial mechanism to this domain. We showed the adjusted mechanisms keep the same properties as their classical versions, namely RPLQ is weakly envy-free and strategy-proof but not ordinally efficient, and PSLQ is ordinally efficient, envy-free, and weakly strategy-proof. We also showed that if we allow quotas to be non-negative real numbers, no random assignment mechanism satisfies ordinal efficiency, envy-freeness and weak strategy-proofness simultaneously for all markets with at least two agents and three projects. Finally, we extended the framework to the random assignment of multiple indivisible objects.

It would be interesting to see how to extend the mechanisms if we do not take feasibility of the assignments for granted, with an available outside option, or if there are additional constraints, such as prerequisites of projects and courses, or incorporating preferences of students over their peers when working in a group. We leave these questions open for future research.

\newpage
\bibliographystyle{apalike} 
\bibliography{bibliography} 
\newpage

\appendix

%%%%%%%%%%%%%%%%%%%%%%%%%%%%%%%%%%%%%%%%%%%%%%%%%%%%%%%%%%%%%%%
\section*{Proofs}
%%%%%%%%%%%%%%%%%%%%%%%%%%%%%%%%%%%%%%%%%%%%%%%%%%%%%%%%%%%%%%%

%%%%%%%%%%%%%%%%%%%%%%%%%%%%%%%%%%%%%%%%%%%%%%%%%%%%%%%%%%%%%%
\section{Proposition~\ref{RPLQ}}\label{proof RPLQ}
%%%%%%%%%%%%%%%%%%%%%%%%%%%%%%%%%%%%%%%%%%%%%%%%%%%%%%%%%%%%%%

%\begin{comment}
\subsubsection*{RPLQ is strategy-proof} 

\begin{proof}
We have that $\forall \sigma \in S_n$, $PrioLQ(\succ)^{\sigma}$ is strategy-proof. Since $RPLQ(\succ)$ is a convex combination of $PrioLQ(\succ)^{\sigma}$ over all $\sigma \in S_n$ and is itself feasible with constant coefficients independent of $\succ$, strategy-proofness is preserved.
\end{proof}

\subsubsection*{RPLQ is weekly envy-free} 

\begin{proof}
Fix $\succ \in \mathcal{P}^n$ and let $R=RPLQ(\succ)$. Suppose $R_2 \ sd(\succ_1) \ R_1$. To prove the claim, we must show $R_2=R_1$. Enumerate the projects in a decreasing order of preference of student 1: $p_1,\dots,p_k$ such that $p_1 \succ_1 p_2 \succ_1 \dots \succ_1 p_k$.

For any permutation $\sigma \in S_n$ in which 1 precedes 2, let $\sigma'=\sigma(12)$, i.e. it is a permutation which we get from $\sigma$ by transposing it by the two-cycle $(12)$. Since the rest of the elements remain fixed, we have that the pairs $\{\sigma,\sigma'\}$ partition the symmetric group $S_n$. Since $\succ$ is fixed throughout this part of the proof, we will suppress it from now until the end of the proof. Define $Q=\frac{PrioLQ^{\sigma}+PrioLQ^{\sigma'}}{2}$.

Let us first consider the allocation of $p_1$. We distinguish 3 cases. These scenarios might follow because we hit an upper or a lower quota, or both at the same time. For the purpose of this proof, however, this makes no difference. Indeed, notice that if at step $k$ we limit the choice of the student $\sigma(k)$ to $P_l^k$,\footnote{Recall $P_l^k \subseteq P$ is the set of projects with unfilled lower quota at the beginning of step $k$.} then none of the students $\sigma(k+1),..,\sigma(n)$ will be able to choose from $P \setminus P_l^k$, even though some of these projects might have unfilled upper quotas. 

The cases are as follows:
\begin{enumerate}
    \item 1 does not get $p_1$ at neither $PrioLQ^{\sigma}$ nor $PrioLQ^{\sigma'}$
    \item 1 gets $p_1$ at $PrioLQ^{\sigma}$ and 2 cannot get $p_1$ at $PrioLQ^{\sigma}$
    \item 1 gets $p_1$ at $PrioLQ^{\sigma}$ and 2 can get $p_1$ at $PrioLQ^{\sigma}$
\end{enumerate}

If case (1) holds, neither student can get $p_1$ at neither $PrioLQ^{\sigma}$ nor $PrioLQ^{\sigma'}$, and so $r_{1p_1}=r_{2p_1}=0$.

Suppose case (2) holds. If 2 gets $p_1$ at $PrioLQ^{\sigma'}$, then so does 1 at $PrioLQ^{\sigma}$. At $PrioLQ^{\sigma}$, 2 cannot get $p_1$. Therefore, $q_{2p_1} \leq q_{1p_1}$.

Now suppose case (3) holds. Then 1 gets $p_1$ in both $PrioLQ^{\sigma}$ and $PrioLQ^{\sigma'}$, so again we have $q_{1p_1} \geq q_{2p_1}$.

Since $R$ is a convex combination of such such assignments, we have $q_{1p_1} \geq q_{2p_1}$ implies $r_{1p_1} \geq r_{2p_1}$. Since we assume $r_{2p_1} \geq r_{1p_1}$, we must have $r_{1p_1}=r_{2p_1}$. Thus, for all such pairs of permutations $\{\sigma,\sigma'\}$, we have that exactly one of the following holds (we will denote these conditions by $(\star)$)
\begin{enumerate}
    \item 1 gets $p_1$ at $PrioLQ^{\sigma}$ but not in $PrioLQ^{\sigma'}$ and 2 gets $p_1$ at $PrioLQ^{\sigma'}$ but not at $PrioLQ^{\sigma}$.
    \item Both get $p_1$ at both assignments.
    \item Neither gets $p_1$ in any $PrioLQ^{\sigma}$ or $PrioLQ^{\sigma'}$.
\end{enumerate}

Next, we consider the allocation of $p_2$. We distinguish the following cases:
\begin{enumerate}
    \item 1 gets $p_2$ at $PrioLQ^{\sigma}$ and 2 cannot get $p_2$ at $PrioLQ^{\sigma}$.
    \item 1 gets $p_2$ at $PrioLQ^{\sigma}$ and 2 can get $p_2$ at $PrioLQ^{\sigma}$.
    \item 1 gets $p_1$ at $PrioLQ^{\sigma}$ and 2 can get $p_2$ at $PrioLQ^{\sigma}$.\footnote{If 1 gets $p_1$ at $PrioLQ^{\sigma}$ and 2 cannot get $p_2$ at $PrioLQ^{\sigma}$, see the paragraph explaining the allocation of $p_1$.}
    \item 1 does not get $p_1$ nor $p_2$ at $PrioLQ^{\sigma}$.
\end{enumerate}

First, consider case (4). Then neither of the students gets $p_1$ nor $p_2$ under $PrioLQ^{\sigma}$ nor $PrioLQ^{\sigma'}$. Therefore, $r_{1p_1}=r_{2p_1}=r_{1p_2}=r_{2p_2}=0$.

Next, consider case (1). 2 cannot get $p_1$ at neither assignments, nor does 1. If 2 gets $p_2$ at $PrioLQ^{\sigma'}$, then 1 cannot get $p_2$ at $PrioLQ^{\sigma'}$.

Case (2) implies that if 1 and 2 get $p_2$ at $PrioLQ^{\sigma}$, then 1 certainly gets $p_2$ at $PrioLQ^{\sigma'}$.

Finally, let us consider case (3). If 1 gets $p_1$ at $PrioLQ^{\sigma}$ and 2 gets $p_2$ at $PrioLQ^{\sigma}$, then 1 gets no worse than $p_2$ at $PrioLQ^{\sigma'}$, so 2 cannot get worse than $p_2$ at $PrioLQ^{\sigma'}$. However, by $(\star)$, we must have that 2 gets $p_1$ at $PrioLQ^{\sigma'}$, so 1 gets $p_2$ in $PrioLQ^{\sigma'}$.

We conclude that $q_{1p_2} \geq q_{2p_2}$, which implies $r_{1p_2} \geq r_{2p_2}$. By our assumption, we have $r_{2p_1}+r_{2p_2} \geq r_{1p_1}+r_{1p_2}$, and since we have already shown $r_{1p_1}=r_{2p_1}$, it must follow that $r_{1p_2} = r_{2p_2}$.

Observe that for all pairs $\{\sigma,\sigma'\}$, we must have that the allocation of $p_1,p_2,P \setminus \{p_1,p_2\}$ satisfies the following for at least one of $x$ and $y$ in $\{p_1,p_2\}$ and $x \succ_1 y$
\begin{equation}\label{eq envy-freeness 1}
    PrioLQ_1^{\sigma} = x \text{ and } PrioLQ_2^{\sigma} = x \Rightarrow PrioLQ_1^{\sigma'} = x \text{ and } PrioLQ_2^{\sigma'} = x
\end{equation}

and 
\begin{equation}\label{eq envy-freeness 2}
    PrioLQ_1^{\sigma} = x \text{ and } PrioLQ_2^{\sigma} = y \Rightarrow PrioLQ_1^{\sigma'} = y \text{ and } PrioLQ_2^{\sigma'} = x
\end{equation}

We proceed by induction. Suppose $r_{1p_i}=r_{2p_i}, \forall i \in \{1,..,l-1\}$. Also suppose that $\forall x,y \in \{p_1,p_2,..,p_{l-1},P \setminus \{p_1,p_2,..,p_{l-1}\}\}$, where at least one of them is in $\{p_1,p_2,..,p_{l-1}\}$ and $x \succ_1 y$, ~\ref{eq envy-freeness 1} and ~\ref{eq envy-freeness 2} hold.

If 2 gets $p_l$ at $PrioLQ^{\sigma'}$, then by the induction hypothesis 1 gets a project from $P \setminus \{p_1,..,p_{l-1}\}$ at $PrioLQ^{\sigma}$. But since $p_l$ is the best for her in this set and it is available, 1 gets $p_l$ at $PrioLQ^{\sigma}$. If 2 gets $p_l$ at $PrioLQ^{\sigma}$, then 1 gets $p_e$, $e \leq l$, at $PrioLQ^{\sigma}$. Then, again by the induction hypothesis, 2 gets $p_e$ at $PrioLQ^{\sigma'}$, and $p_l$ must be available for 1 at $PrioLQ^{\sigma'}$. But this implies 1 must get $p_l$ at $PrioLQ^{\sigma'}$.

We conlcude that $q_{2p_l} \leq q_{1p_l}$, so $r_{2p_l} \leq r_{1p_l}$. Since $\sum_{i=1}^{l} r_{2p_i} \geq \sum_{i=1}^{l} r_{1p_i}$ by assumption and $r_{1p_i}=r_{2p_i}, \forall i \in \{1,..,l-1\}$ by the induction hypothesis, it follows that $r_{2p_l} = r_{1p_l}$. The result follows by the principle of mathematical induction.
\end{proof}

%%%%%%%%%%%%%%%%%%%%%%%%%%%%%%%%%%%%%%%%%%%%%%%%%%%%%%%%%%%%%%%%%%
\section{Lemma~\ref{lemma: characterisation ordinal efficiency}}\label{proof lemma: characterisation ordinal efficiency}
%%%%%%%%%%%%%%%%%%%%%%%%%%%%%%%%%%%%%%%%%%%%%%%%%%%%%%%%%%%%%%%%%%

Fix $R$ and $\succ$ as in the statement of the lemma. First, we prove that if $\tau(R,\succ)$ is cyclic or $R$ is wasteful at $\succ$, $R$ is not ordinally efficient.

\subsubsection*{Statement only if}

\begin{proof}
First, suppose $R$ contains a wasteful chain at $\succ$, i.e. there is a sequence of students $i_1,\dots,i_l \in N$ and a sequence of projects $p_1,\dots,p_{l+1}$ such that all of the following are satisfied: (i) $p_j \succ_{i_j} p_{j+1}$ and (ii) $r_{i_jp_{j+1}} > 0, \forall 1 \leq j \leq l$; and (iii) $\sum_{i \in N} r_{ip_1} < u(p_1)$ and $\sum_{i \in N} r_{ip_{l+1}} > l(p_{l+1})$.

Define 
\begin{equation*}
    \delta = \min \{r_{i_1p_2},\dots,r_{i_lp_{l+1}},u(p_1) - \sum_{i \in N}r_{ip_1},\sum_{i \in N}r_{ip_{l+1}} - l(p_{l+1})\}
\end{equation*}

In particular, $\delta > 0$.

Define $R' = R + \Delta$, where $\Delta$ is an $n$-by-$k$ matrix defined as follows: $\Delta_{i_jp_j} = \delta$, $\Delta_{i_jp_{j+1}} = -\delta$ and 0 otherwise, $\forall j \in \{1,\dots,l+1\}$. By construction, $R$ is feasible $\Rightarrow$ $R'$ is also feasible. Moreover, $R_j' \ sd(\succ_j) \ R_j$ and $R_j \neq R_j'$ $\forall j \in \{1,\dots,l+1\}$, while $R_i = R_i', \forall j\neq i \in N$. Thus, $R'$ stochastically dominates $R$ with respect to $\succ$, so $R$ is ordinally inefficient.

Secoondly, suppose $\tau(R,\succ)$, which we denote shorthand by $\tau$ for simplicity, contains a cycle 

\begin{equation*}
p_1 \ \tau \ p_2 \ \tau \ \dots \ \tau \ p_l \ \tau \ p_1
\end{equation*}

We assume without a loss of generality that projects $p_1$ through $p_l$ are all distinct. By definition of $\tau$, there is a sequence of students (not necessarily all distinct) $i_1,\dots,i_l \in N$ such that

\begin{equation*}
    r_{i_1p_1} > 0 \text{ and } p_2 \succ_{i_1} p_1
\end{equation*}

\begin{equation*}
    r_{i_2p_2} > 0 \text{ and } p_3 \succ_{i_2} p_2
\end{equation*}

$$\vdots$$

\begin{equation*}
    r_{i_lp_l} > 0 \text{ and } p_1 \succ_{i_l} p_l
\end{equation*}

Take $\delta = \min \{r_{i_1p_1},\dots,r_{i_lp_l} \}$; note that $\delta>0$. Define $R'=R + \Delta$ where $\Delta_{i_jp_j} = -\delta$ and $\Delta_{i_jp_{j+1}} = \delta$, for $j=1,\dots,l$; and $\Delta_{ip}=0$, otherwise. By construction, $R'$ is a row stochastic matrix which is feasible if $R$ is feasible. Finally, $R'$ stochastically dominates $R$ since $\forall i \in \{i_1,\dots,i_l\}$, $R_i' \ sd(\succ_i) \ R_i$ and $R_i' \neq R_i$ (if the same student appears more than once in the list $\{i_1,\dots,i_l\}$, we use transitivity of the stochastic dominance relation).
\end{proof}

Next, we prove that if $R$ is not ordinally efficient at $\succ$, it is either wasteful at $\succ$ or $\tau(R,\succ)$ is cyclic.

\subsubsection*{Statement if}

\begin{proof}
Suppose $R$ is stochastically dominated at $\succ$ by $R'$. By definition, $R_i' \ sd(\succ_i) \ R_i, \forall i \in N$ and $\exists i_1 \in N$ such that $R_{i_1}' \neq R_{i_1}$. By definition of $sd(\succ_{i_1})$, this implies that $\exists p,q \in P$ such that $p \succ_{i_1} q, r_{i_1q}'< r_{i_1q}$, and $r_{i_1p}' > r_{i_1p}$. Note that this implies $r_{i_1q}>0$, and in particular $p \ \tau(R,\succ) \ q$.

We distinguish four cases:

\smallskip

\textbf{Case 1:} $\sum_{i \in N}r_{ip} < u(p)$ and $\sum_{i \in N}r_{iq} > l(q)$. Then $R$ is wasteful at $\succ$.

\smallskip

\textbf{Case 2:} $\sum_{i \in N}r_{ip} = u(p)$ and $\sum_{i \in N}r_{iq} > l(q)$. By feasibility, $\exists i_2 \in N$ such that $r_{i_2p}'<r_{i_2p}$. Since $R_{i_2}' \ sd(\succ_{i_2}) \ R_{i_2}$ and now that we showed $R_{i_2}' \neq R_{i_2}$, there must exist $x \in P$ such that $x \succ_{i_2} p$ and $r_{i_2x}'>r_{i_2x}$; in particular, $x \ \tau(R,\succ) \  p$. Repeating this argument, since the sets $N$ and $P$ are finite, either we find a project $y \in P$ such that $\sum_{i \in N}r_{iy} < u(y)$, which would constitute the "upper end" of a wasteful chain, or a cycle in the relation $\tau(R,\succ)$.\footnote{Note that if $\sum_{i \in N}r_{iy} = u(y), \forall y \in P$, we find ourselves in the setting of \cite{bogomolnaia2001new}, where acyclicity of $\tau(R,\succ)$ is both sufficient and necessary for ordinal efficiency.}

\smallskip

\textbf{Case 3:} $\sum_{i \in N}r_{ip} < u(p)$ and $\sum_{i \in N}r_{iq} = l(q)$. By feasibility, $\exists i_2 \in N$ such that $r_{i_2q}'>r_{i_2q}$. Since $R_{i_2}' \ sd(\succ_{i_2}) \ R_{i_2}$ and now that we showed $R_{i_2}' \neq R_{i_2}$, there must exist $x \in P$ such that $q \succ_{i_2} x$ and $r_{i_2x}'<r_{i_2x}$; in particular, $q \ \tau(R,\succ) \  x$. Similarly to the argument above, repeating this, since the sets $N$ and $P$ are finite, either we find a project $y \in P$ such that $\sum_{i \in N}r_{iy} > l(y)$, which would constitute the "lower end" of a wasteful chain, or a cycle in the relation $\tau(R,\succ)$.\footnote{Similarly to the previous footnote, if $\sum_{i \in N}r_{iy} = l(y), \forall y \in P$, acyclicity of $\tau(R,\succ)$ is both sufficient and necessary for ordinal efficiency.}

\smallskip

\textbf{Case 4:} $\sum_{i \in N}r_{ip} = u(p)$ and $\sum_{i \in N}r_{iq} = l(q)$. We repeat the arguments used in cases 2 and 3. By iteration, either we find an upper and a lower end of a wasteful chain, or the random assignment is equal to (i) the upper quota at each project or (ii) the lower quota at each project. In either case, the relation $\tau(R,\succ)$ must be cyclic.
\end{proof}

%%%%%%%%%%%%%%%%%%%%%%%%%%%%%%%%%%%%%%%%%%%%%%%%%%%%%%%%%%%%%%%%%%
\section{Proposition~\ref{PSLQ o-efficient}}\label{proof PSLQ o-efficient}
%%%%%%%%%%%%%%%%%%%%%%%%%%%%%%%%%%%%%%%%%%%%%%%%%%%%%%%%%%%%%%%%%%

We complete the proof in two steps, proving each statement in turn.

\subsubsection*{PSLQ is ordinally efficient}

\begin{Lemma}\label{lemma: PSLQ no wasteful chain}
For any market $(N,P,{l},{u},\succ)$, PSLQ does not contain a wasteful chain.
\end{Lemma}

\begin{proof}
We use a proof by contradiction. Fix $(N,P,{l},{u},\succ)$, denote the output of PSLQ by $R$ and by $i_1,\dots,i_l$ and $p_1,\dots,p_{l+1}$ a wasteful chain. By definition of a wasteful chain, we have $\sum_{i \in N}r_{ip_1}<u(p_1)$, which implies that students must be shifted from $p_1$; let us denote this time by $t_c$ and let us write $\mathcal{A}$ shorthand for $\mathcal{A}(t_c)$, the set of active projects at $t_c$. Since $p_1 \succ_{i_1} p_2$ and $r_{i_1p_2}>0$, $p_2 \in \mathcal{A}$. Similarly, if $\{p_2,\dots,p_{n}\} \subseteq \mathcal{A}$, $p_{n} \succ_{i_n} p_{n+1}$ and $r_{i_np_{n+1}}>0$ for some $n \in \{2,\dots,l\}$, it follows that $p_{n+1} \in \mathcal{A}$. By induction, we find $p_j \in \mathcal{A}, \forall j=2,\dots,l+1$. By Equations~\ref{eq: once shifted the rest has lower quota 1} and~\ref{eq: once shifted the rest has lower quota 2} and the discussion that follows them, all projects in $\mathcal{A}$ must have their total fractional assignment equal to their lower quota. However, this contradicts $\sum_{i \in N}r_{ip_{l+1}} > l(p_{l+1})$.
\end{proof}

We are ready to prove the PSLQ mechanism is ordinally efficient.

\begin{proof}
We prove this claim by contradiction. Fix $(N,P,{l},{u},\succ
)$ and suppose $R=PSLQ(\succ)$ is not ordinally efficient. By Lemma~\ref{lemma: PSLQ no wasteful chain}, $R$ does not contain a wasteful chain at $\succ$, so by Lemma~\ref{lemma: characterisation ordinal efficiency}, we can find a cycle in the relation $\tau(R,\succ)$: $p_1 \ \tau \ p_2 \ \tau \ \dots \ \tau \ p_l \ \tau \ p_1$. Let $i_s \in N$ be a student such that $p_{s-1}\succ_{i_s}p_s$ and $r_{i_sp_s}>0$ ($s \in \{1,..,l\}$ where $p_{l+1}=p_1$). Let $v^s$ be the first step $v$ in the PSLQ eating algorithm when the student $i_s$ starts to acquire project $p_s$, that is the last step for which $r_{i_sp_s}^v \neq 0$. Since $p_{s-1}$ is strictly preferred by student $i_s$ and $i_s$ gets a positive amount of $p_s$, $p_{s-1}$ must be either fully distributed or shifted from so that all projects satisfy their lower quotas. But since if a student is shifted from a project, none of the students will be able to eat from it for the rest of the execution of the algorithm, in both cases (fully distributed or shifted), we must have $v^{s-1}<v^s$ and this holds $\forall s \in \{1,..,l\}$. Hence, $v^1<v^2< \dots <v^l<v^1$, a contradiction.
\end{proof}

\subsubsection*{PSLQ is envy-free}

\begin{proof}
Fix $(N,P,{l},{u},\succ)$ and $i \in N$. Label $P$ such that $p_1 \succ_i p_2 \succ_i \dots \succ_i p_k$. Let $s_1$ be the step at which $p_1$ is fully assigned or shifted from, that is $p_1 \in \mathcal{A}^{s_1-1} \setminus \mathcal{A}^{s_1}$. Because $p_1 \in \mathcal{A}^v$ for all $v \leq s_1-1$, we have

\begin{equation*}
r_{ip_1}^{s_1} = t^{s_1} \geq  r_{jp_1}^{s_1}, \forall j \in N    
\end{equation*}

Suppose $p_1$ is fully allocated at step $s_1$. Then $r_{ip_1}^{s_1} = r_{ip_1}$ and $r_{jp_1}^{s_1} = r_{jp_1}, \forall j \in N$, that is they are the i-th and j-th entry of the resulting random assignment matrix $R=PSLQ(\succ)$. Now suppose 1 is shifted from $p_1$ at $s_1$. Since all students eating from $P$ at $s_1$ must shift too, we obtain the same conclusion, i.e. $r_{ip_1}^{s_1} = r_{ip_1}$ and $r_{jp_1}^{s_1} = r_{jp_1}$. Therefore,

\begin{equation*}
r_{ip_1} \geq r_{jp_1}, \forall j \in N    
\end{equation*}

We proceed by induction. Suppose 

\begin{equation*}
\sum_{l=1}^{m-1}r_{ip_l} \geq \sum_{l=1}^{m-1}r_{jp_l}, \forall j \in N, \text{ for } k > m-1 \geq 1    
\end{equation*}

Let $s_m$ be the step at which $\{p_1,..,p_m\}$ are all either fully allocated or shifted from. We have

\begin{equation*}
\{p_1,..,p_m\} \cap \mathcal{A}^{s_m-1} \neq \emptyset \text{ , and } \{p_1,..,p_m\} \cap \mathcal{A}^{s_m} = \emptyset    
\end{equation*}

Note that $s_1 \leq s_2 \leq \dots \leq s_m$. Hence, $\exists 1 \leq h \leq m$ such that $p_h \in \mathcal{A}^v$ for all $v \leq s_m - 1$. It follows that

\begin{equation*}
\sum_{l=1}^{m}r_{ip_l} = \sum_{l=1}^{m}r_{ip_l}^{s_m} = t^{s_m} \geq \sum_{l=1}^{m}r_{jp_l}^{s_m} = \sum_{l=1}^{m}r_{jp_l}, \forall j \in N    
\end{equation*}

Therefore, by induction it follows that 
$$\sum_{l=1}^{m}r_{jp_l} \geq \sum_{l=1}^{m}r_{jp_l}, \forall j \in N, \forall m=1,..,k$$
that is $R_i \ sd(\succ_i) \ R_j, \forall j \in N$.
\end{proof}

%%%%%%%%%%%%%%%%%%%%%%%%%%%%%%%%%%%%%%%%%%%%%%%%%%%%%%%%%%%%%%
\section{Theorem~\ref{PSLQ weakly strategy-proof}}\label{proof PSLQ weakly strategy-proof}
%%%%%%%%%%%%%%%%%%%%%%%%%%%%%%%%%%%%%%%%%%%%%%%%%%%%%%%%%%%%%%

In this section, we prove the PSLQ mechanism is weakly strategy-proof. In what follows, assume market $(N,P,{l},{u},\succ)$ is fixed. 

Before we proceed, let us introduce some notation. Fix $i \in N$, and let $e_i:[0,1] \rightarrow P$ be the eating schedule of student $i$ where $e_i(t)$ is the project student $i$ is eating at time $t$. We require $e_i$ to be right-continuous with respect to the discrete topology on $P$, that is $\forall t \in [0,1), \exists \epsilon >0$ such that $\forall s \in [t,t+\epsilon), e_i(s)=e_i(t)$. For any vector of eating schedules of the students $\mathbf{e} = (e_1,\dots,e_n)$, time $t \in [0,1]$ and project $p \in P$, define
$$N_p(t,\mathbf{e}) = \{i \in N: e_i(t)=p\} \text{ and } n_p(t,\mathbf{e}) = |N_p(t,\mathbf{e})|$$
Let $t_p(\mathbf{e})$ be the time at which $p$ is ceased to be eaten from, that is
$$t_p(\mathbf{e})=\sup\{t \in [0,1]: n_p(t,\mathbf{e}) \geq 1\}$$
Note that $n_p(t,\mathbf{e})$ is a non-decreasing step-function in $t$ on $[0,t_p(\mathbf{e}))$, as once a student starts eating from $p$, it eats from it until all project seats are assigned or until all students are shifted from it.\footnote{A project can be neither shifted from nor exhausted when it is eaten from until the end of the epoch and does not reach its upper quota. However, this distinction will make no difference in what follows and we will omit it for brevity of arguments.} Also recall $\chi_{ip}(t)$ is the inidicator function with $\chi_{ip}(t)=1 \iff e_i(t) = p$. In particular, both $n_p(t,\mathbf{e})$ and $\chi_{ip}(t)$ are Riemann integrable. For student $i$, project $p$ and time $t$, we define
$$ \omega_{ip}(t,\mathbf{e}) = \int_{0}^{t}\chi_{ip}(s) \,ds$$
and
$$ \omega_p(t,\mathbf{e}) = \int_{0}^{t}n_p(s,\mathbf{e})\,ds = \sum_{i \in N} \omega_{ip}(t,\mathbf{e})$$
Note that $\omega_p(t_p(\mathbf{e}),\mathbf{e}) = \omega_p(1,\mathbf{e})$. We also define $\tau_p(\mathbf{e})$ to be the time at which students start eating from $p$, that is
$$\tau_p(\mathbf{e})=\inf\{t \in [0,1]: n_p(t,\mathbf{e}) \geq 1\}$$
and by  $\tau_{ip}(\mathbf{e})$ the first time $i$ eats from $p$ under the eating schedule $\mathbf{e}$. Formally,
$$\tau_{ip}(\mathbf{e}) = \inf \{t \in [0,1]: e_i(t)=p \}$$

We denote by $t^c(\mathbf{e})$ the critical time, i.e.

\begin{equation*}
    t^c(\mathbf{e}) = \inf \{t \in [0,1]: n(1-t) = \sum_{p \in P} \Big( l(p) - \omega_p(t,\mathbf{e}) \Big)_{+} \}
\end{equation*}

We say project $p$ is shifted from if $1>t_p(\mathbf{e}) \geq t^c(\mathbf{e})$.

\begin{comment}
and let $\tau_{ip}(\mathbf{e})$ be the time when student $i$ starts eating from project $p$, i.e.
$$\tau_{ip}(\mathbf{e})=\inf\{t \in [0,1]: e_i(t)=p\}$$
\end{comment}

For any preference profile $\succ$, let $\mathbf{e}^{\succ}$ be the vector of eating schedules generated by PSLQ when students report $\succ$. In all of the following lemmata, we fix $i \in N$ and suppose $i$ misreports to $\succ_i'$. We denote the resulting profile by $\succ'=(\succ_i',\succ_{-i})$. We denote by $R$ the output of $PSLQ(\succ)$ and by $R'$ the output of $PSLQ(\succ')$. Also, we write $t_p$ ($\tau_p$) and $t_p'$ ($\tau_p'$) shorthand for $t_p(\mathbf{e}^{\succ})$ ($\tau_p(\mathbf{e}^{\succ})$) and $t_p(\mathbf{e}^{\succ'})$ ($\tau_p(\mathbf{e}^{\succ'})$), respectively.

\begin{comment}
\begin{Lemma}\label{lemma: proof PSLQ ws 1}
Let $\mathbf{x},\mathbf{y} \in \mathbb{R}_{\geq 0}^k$ be stochastic vectors such that $\mathbf{x} \neq \mathbf{y}$ and suppose $\mathbf{x}$ first-order stochastically dominates $\mathbf{y}$. Then $\exists s \in \{1,\dots,k\}$ such that $x_s > y_s$ and $x_t=y_t, \forall t < s$.
\end{Lemma}

\begin{proof}
Since $\mathbf{x}$ first-order stochastically dominates $\mathbf{y}$, we have by definition

\begin{equation}\label{eq: lemma proof PSLQ ws 1}
\forall l=1,\dots,k: \sum_{j=1}^{l} x_j \geq \sum_{j=1}^{l} y_j    
\end{equation}

with at least one of the inequalities strict. Suppose that $\not\exists s \in \{1,\dots,k\}$ such that $x_s>y_s$. But then $\sum_{j=1}^{k} x_j \leq \sum_{j=1}^{k} y_j$, a contradiction. So there must exist $1 \leq j \leq k$ such that $x_j > y_j$; let $s$ be smallest such integer. If $s=1$, we are done. Suppose $s>1$, i.e. $y_1 \geq x_1$. But by~\ref{eq: lemma proof PSLQ ws 1}, $x_1 \geq y_1$, so we get $x_1 = y_1$. If $s=2$, we are done; otherwise $s>2$, i.e. $y_2 \geq x_2$, and coupling this with~\ref{eq: lemma proof PSLQ ws 1} and the fact that $y_1=x_1$, we obtain $x_2=y_2$. Continuing in this inductive procedure, the result follows.
\end{proof}
\end{comment}

\begin{Lemma}\label{lemma: PSLQ ws helpful can't have more}
Fix $p \in P$ such that $t_{p}' \geq t_{p}$. Suppose $N_p(t,\mathbf{e}^{\succ}) \subseteq N_p(t,\mathbf{e}^{\succ'}), \forall t$ satisfying $\tau_{p} \leq t < t_{p}$. Then $t_{p}' = t_{p}$.
\end{Lemma}

\begin{proof}
Suppose, for the sake of reaching a contradiction, that $t_{p}' > t_{p}$. We distinguish two cases.

\smallskip

\textbf{Case 1:} $\omega_p(1,\mathbf{e}^{\succ}) = u(p)$. If $N_p(t,\mathbf{e}^{\succ}) \subseteq N_p(t,\mathbf{e}^{\succ'}), \forall t \in [\tau_{p},t_{p})$, together with the fact that $\omega_p(t,\mathbf{e}^{\succ})$ is non-decreasing in $t$ and our assumption $t_{p} < t_{p}'$, we obtain $\omega_p(1,\mathbf{e}^{\succ'}) > \omega_p(1,\mathbf{e}^{\succ}) = u(p)$, a contradiction.

\smallskip

\textbf{Case 2:} Students are shifted from $p$ at $t_{p}$, i.e.

\begin{equation*}
n(1-t_{p}) = \sum_{q \in P} \Big( l(q) - \omega_q(t_p,\mathbf{e}^{\succ}) \Big)_+ = \sum_{q \in P \setminus \{p\}} \Big( l(q) - \omega_q(t_p,\mathbf{e}^{\succ}) \Big)_+ 
\end{equation*}

Now suppose $N_p(t,\mathbf{e}^{\succ}) \subseteq N_p(t,\mathbf{e}^{\succ'}), \forall t \in [\tau_{p},t_{p})$, together with the fact that $\omega_p(t,\mathbf{e}^{\succ})$ is non-decreasing in $t$. This implies $\omega_p(t_p,\mathbf{e}^{\succ'}) \geq \omega_p(t_p,\mathbf{e}^{\succ})$ and hence

\begin{equation*}
n(1-t_{p}) = \sum_{q \in P} \Big( l(q) - \omega_q(t_p,\mathbf{e}^{\succ}) \Big)_+ \leq \sum_{q \in P} \Big( l(q) - \omega_q(t_p,\mathbf{e}^{\succ'}) \Big)_+
\end{equation*}

Finally, since $l(p) \leq \omega_p(t_p,\mathbf{e}^{\succ}) \leq \omega_p(t_p,\mathbf{e}^{\succ'})$, $p$ has to cease to be active at $t_p$ under $\succ'$, a contradiction to $t_p < t_p'$.
\end{proof}

\begin{Lemma}\label{lemma: PSLQ ws others continue eating}
For any $p \in P$, $t \in [0,\min\{t_{p},t_p'\})$ and $j \in N \setminus \{i\}$,
$$j \in N_p(t,\mathbf{e}^{\succ}) \Rightarrow j \in N_p(t,\mathbf{e}^{\succ'})$$
\end{Lemma}

\begin{proof}
Suppose there is student $i\neq j \in N$ and time $0 \leq t < \min\{t_{p},t_p'\}$ such that $j \in N_p(t,\mathbf{e}^{\succ})$ and $j \in N_q(t,\mathbf{e}^{\succ'})$ for some $q \neq p$. Observe that project $p$ is available at $t <  t_{p}'$ under $\succ'$, so we must have $q \succ_j p$. It follows that $q$ is not available at $t$ under $\succ$, as otherwise $j$ would be eating it at $t$. In other words, $t_{q} \leq t < t_{q}'$, regardless of whether students are shifted from or hit the upper quota of $q$ at $t_{q}$.

Let $S \subset P$ be the set of projects $q \neq p$ such that $t_{q} < t_{q}'$. We showed $S \neq \emptyset$ in the paragraph above. Define

\begin{equation*}
x = \argmin_{q \in S} t_{q}    
\end{equation*}

It follows from $t_{x} < t_{x}'$ that there is some time $t<t_{x}$ and a student $k$ such that $k \in N_x(t,\mathbf{e}^{\succ})$ and $k \notin N_x(t,\mathbf{e}^{\succ'})$. To see why this holds, we suppose to the contrary and apply Lemma~\ref{lemma: PSLQ ws helpful can't have more} to obtain a contradiction. Denote by $y$ the project satisfying $k \in N_y(t,\mathbf{e}^{\succ'})$. 

As project $x$ is available at $t$ under $\succ$, we must have $y \succ_k x$. But recall that student $k$ eats $x$ at $t$ under $\succ$, and $k \neq i$, so the preferences of $k$ do not change between the two profiles. We therefore deduct that $y$ is not available at $t$ under $\succ$. Thus, we have shown

\begin{equation*}
    t < t_{x} \text{, } t < t_{y}' \text{, and } t_{y} < t
\end{equation*}

This in turn implies $t_{y} < t_{y}'$, and so $y \in S$ with $t_{y} < t_{x}$, contradicting the minimality property of $x$.

\end{proof}

\begin{comment}
\begin{Lemma}
Fix $i \in N$, and let $p \in P$ be the most favourite project of $i$ under $\succ_i$. Suppose $i$ misreports to $\succ_i'$. Then for any project $q \in P \setminus \{p\}$, we have

\begin{equation*}
    \omega_q(t,\mathbf{e}^{\succ}) \leq \omega_q(t,\mathbf{e}^{\succ'}), \forall t \in [0,1]
\end{equation*}
\end{Lemma}

\begin{proof}

\end{proof}
\end{comment}

\begin{Lemma}\label{lemma: PSLQ ws if i 0 then under misreport will not ever decrease}
Fix $p=\phi_i(P)$ and suppose $t^c \leq t^{c'}$. Then for all $q \in P \setminus \{p\}$ such that $r_{iq}=0$ and $q \notin \mathcal{A}(t^c)$, $\omega_q(t,\mathbf{e}^{\succ'}) \geq \omega_q(t,\mathbf{e}^{\succ}), \forall t \in [0,1]$.
\end{Lemma}

\begin{proof}
Fix $q \in P \setminus \{p\}$ such that $r_{iq}=0$ and $q \notin \mathcal{A}(t^c)$. Since the preferences of all agents other than $i$ do not change between the profiles, $t^c \leq t^{c'}$ by assumption and since $i$ did not eat from $q$ at all under $\succ'$, it follows that $N_q(t,\mathbf{e}^{\succ}) \subseteq N_q(t,\mathbf{e}^{\succ'}), \forall t \in [0,\min\{t_q,t_q'\})$ by Lemma~\ref{lemma: PSLQ ws others continue eating}. If $t_q' < t_q$, since $t_q \leq t^c \leq t^{c'}$, it follows that $\omega_q(t_q',\mathbf{e}^{\succ'}) = u(q)$ and the result follows for all $t \in [0,1]$ by feasibility. If, on the other hand, we have $t_q' \geq t_q$, it follows that $\omega_q(t_q,\mathbf{e}^{\succ'}) \geq  \omega_q(t_q,\mathbf{e}^{\succ})$. Since $\omega_q$ is a non-decreasing function in $t$, it follows $\omega_q(t,\mathbf{e}^{\succ'}) \geq  \omega_q(t,\mathbf{e}^{\succ}), \forall t \in [0,1]$.
\end{proof}

\begin{Lemma}\label{lemma: PSLQ ws if i 0 then under misreport will not decrease}
Fix $p=\phi_i(P)$ and suppose $t^c \leq t^{c'}$. Then for all $q \in P \setminus \{p\}$ such that $r_{iq}=0$, $\omega_q(1,\mathbf{e}^{\succ'}) \geq \omega_q(1,\mathbf{e}^{\succ})$.
\end{Lemma}

\begin{proof}
By feasibility, if $\omega_q(1,\mathbf{e}^{\succ})=l(q)$, we must have $\omega_q(1,\mathbf{e}^{\succ'}) \geq \omega_q(1,\mathbf{e}^{\succ})$. Suppose $q \in P \setminus \{p\}$ with $r_{iq}=0$ and $u(q) \geq \omega_q(1,\mathbf{e}^{\succ})>l(q)$. Note that this implies $t_q \leq t^c$. The result follows by Lemma~\ref{lemma: PSLQ ws if i 0 then under misreport will not ever decrease}.

\begin{comment}
Moreover, suppose $\omega_q(1,\mathbf{e}^{\succ'}) < \omega_q(1,\mathbf{e}^{\succ})$. Since the preferences of all agents other than $i$ do not change between the profiles and $t^c \leq t^{c'}$ by assumption and since $i$ did not eat from $q$ at all under $\succ'$, it follows that $N_q(t,\mathbf{e}^{\succ}) \subseteq N_q(t,\mathbf{e}^{\succ'}), \forall t \in [0,\min\{t_q,t_q'\})$. Together with $\omega_q(1,\mathbf{e}^{\succ'}) < \omega_q(1,\mathbf{e}^{\succ})$, this implies that $t_q > t_q'$. Finally, since $\omega_q(1,\mathbf{e}^{\succ'}) < \omega_q(1,\mathbf{e}^{\succ}) \leq u(q)$, students must shift from $q$ at $t_q'$, which implies $t^{c'} \leq t_q(,\mathbf{e}^{\succ'}) < t_q \leq t^c \leq t^{c'}$, a contradiction.
\end{comment}
\end{proof}

\begin{Lemma}\label{lemma: PSLQ ws cosntraints tight before then lower anyway}
Fix $p=\phi_i(P)$ and suppose $t^c < t_p$. Then $\omega_{p}(1,\mathbf{e}^{\succ'}) = l(p)$.
\end{Lemma}

\begin{proof}
First, note that since $t^c < t_p$, we must have $\omega_{p}(1,\mathbf{e}^{\succ}) = l(p)$. Suppose, for the sake of reaching a contradiction, that $\omega_{p}(1,\mathbf{e}^{\succ'}) > l(p)$. But then $t^{c'} \geq t_{p}' > t_{p} > t^c$. This implies that $\forall q \in P \setminus \{p\}$, $\omega_q(1,\mathbf{e}^{\succ'}) \geq \omega_q(1,\mathbf{e}^{\succ})$. Indeed, if $r_{iq}>0$ for some $q \in P \setminus \{p\}$, it must hold that $t_q>t^c$, so in particular $\omega_q(1,\mathbf{e}^{\succ})=l(q)$. By feasibility, $\omega_q(1,\mathbf{e}^{\succ}) \geq \omega_q(1,\mathbf{e}^{\succ})=l(q)$. The other case is covered by Lemma~\ref{lemma: PSLQ ws if i 0 then under misreport will not decrease}.

But then

\begin{equation*}
    n = \sum_{x \in P} \omega_x(1,\mathbf{e}^{\succ}) < \sum_{x \in P} \omega_x(1,\mathbf{e}^{\succ'}) = n
\end{equation*}

a contradiction.
\end{proof}

\begin{Lemma}\label{lemma: PSLQ ws shifted structure vector}
Enumerate $P$ as $p_1 \succ_i p_2 \succ \dots \succ p_k$. Suppose $i$ is shifted from $p_l$ for some $1 \leq l < k$ at $t_{p_l}$. Then for all $l < j \leq k$ such that $p_j \notin \mathcal{A}(t_{p_l})$, we must have $r_{ip_j} = 0$.
\end{Lemma}

\begin{proof}
This follows directly from Definition~\ref{def active project}. Since any such $p_j$ as considered in the statement is not among the active projects at $t_{p_l}$, and since no inactive project can turn active again, no student can eat from $p_j$ for the rest of the execution of the algorithm.
\end{proof}

\begin{Lemma}\label{lemma: PSLQ ws shifted come earlier}
Fix $p=\phi_i(P)$ and suppose $t^{c'} \leq t^c < t_p < 1$. Then $\tau_{jp}' \leq \tau_{jp}, \forall j \in N \setminus \{i\}$.
\end{Lemma}

\begin{proof}
Suppose $\exists j \in N \setminus \{i\}$ such that  $\tau_{jp}' > \tau_{jp}$. Note that since $i$ misreports, $\tau_{ip}\geq \tau_{ip}'$. Since $t^{c'} \leq t^c$ and $p \in \mathcal{A}(t^c)$, it follows that $p \in \mathcal{A}(t^{c'})$. Hence, since $t_p < 1$, students must be eventually shifted to $p$. If $\mathcal{A}(t^{c'}) \subseteq \mathcal{A}(t^c)$, $\tau_{jp}' > \tau_{jp}$ would contradict $\succ_{j}=\succ_{j}'$. So there must exist $q \in \mathcal{A}(t^{c'}) \setminus \mathcal{A}(t^c)$. Since this implies $\omega_q(1,\mathbf{e}^{\succ'})=l(q) < \omega_q(1,\mathbf{e}^{\succ})$, it follows that $\exists y \in \mathcal{A}(t^c)$ such that $\omega_y(t^{c'},\mathbf{e}^{\succ'}) > l(y) = \omega_y(t^c,\mathbf{e}^{\succ})$ as we allocate the total mass of $n$.

We claim $\forall z \in P \setminus \mathcal{A}(t^c), \omega_z(t,\mathbf{e}^{\succ'}) \geq \omega_z(t,\mathbf{e}^{\succ}), \forall t \in [0,t^{c'})$. Indeed, fix arbitrary $z \in P \setminus \mathcal{A}(t^c)$. Since the preferences of all agents other than $i$ do not change between the profiles and $i$ did not eat from $z$ at all under $\succ'$, it follows that $N_z(t,\mathbf{e}^{\succ}) \subseteq N_z(t,\mathbf{e}^{\succ'}), \forall t \in [0,\min\{t_z,t_z',t^{c'}\})$ by Lemma~\ref{lemma: PSLQ ws others continue eating}. Hence, $\omega_z(t,\mathbf{e}^{\succ'}) \geq \omega_z(t,\mathbf{e}^{\succ}), \forall t \in [0,t^{c'})$, as required.

This implies

\begin{equation*}
    \sum_{x \in P \setminus \mathcal{A}(t^c)} \omega_x(t^{c'},\mathbf{e}^{\succ'}) \geq \sum_{x \in P \setminus \mathcal{A}(t^c)} \omega_x(t^{c'},\mathbf{e}^{\succ})
\end{equation*}
and hence
\begin{equation}\label{eq: PSLQ ws lemma total to prime at most original}
    \sum_{x \in \mathcal{A}(t^c)} \omega_x(t^{c'},\mathbf{e}^{\succ'}) \leq \sum_{x \in \mathcal{A}(t^c)} \omega_x(t^{c'},\mathbf{e}^{\succ})
\end{equation}

\begin{comment}
Putting this together, the preceding paragraphs imply that the increase in $\omega_y$ to $\omega_y(t^{c'},\mathbf{e}^{\succ'}) = u(q)$ must be achieved by reallocating students within the set $\mathcal{A}(t^c)$. Since $\omega_y$ is non-decreasing in $t$ and since $y \in \mathcal{A}(t^c)$, we must have

\begin{equation*}
    \omega_y(t^{c'},\mathbf{e}^{\succ}) \leq \omega_y(t^c,\mathbf{e}^{\succ}) < l(y) \text{ and } \omega_y(t^{c'},\mathbf{e}^{\succ'}) \leq \omega_y(t^{c'},\mathbf{e}^{\succ}) + t^{c'}
\end{equation*}

But then

\begin{equation*}
    \omega_y(t^{c'},\mathbf{e}^{\succ'}) \leq \omega_y(t^{c'},\mathbf{e}^{\succ}) + t^{c'} < l(y) + t^{c'} \leq l(y) + 1 \leq u(y)
\end{equation*}
which contradicts $\omega_y(t^{c'},\mathbf{e}^{\succ'}) = u(y)$.
\end{comment}

By definition of $t^{c'}$ and since $t^c \geq t^{c'}$,
\begin{equation*}
    \sum_{x \in P} \Big( l(x) - \omega_x(t^{c'},\mathbf{e}^{\succ'}) \Big)_{+} = n(1-t^{c'}) \geq \sum_{x \in P} \Big( l(x) - \omega_x(t^{c'},\mathbf{e}^{\succ}) \Big)_{+}
\end{equation*}
Since
\begin{equation*}
    \sum_{x \in P \setminus \mathcal{A}(t^{c'})} \Big( l(x) - \omega_x(t^{c'},\mathbf{e}^{\succ'}) \Big)_{+} \leq \sum_{x \in P \setminus \mathcal{A}(t^{c'})} \Big( l(x) - \omega_x(t^{c'},\mathbf{e}^{\succ}) \Big)_{+}
\end{equation*}
we must have
\begin{equation*}
    \sum_{x \mathcal{A}(t^{c'})} \Big( l(x) - \omega_x(t^{c'},\mathbf{e}^{\succ'}) \Big)_{+} \geq \sum_{x \in \mathcal{A}(t^{c'})} \Big( l(x) - \omega_x(t^{c'},\mathbf{e}^{\succ}) \Big)_{+}
\end{equation*}

On the other hand, since we have shown Equation~\ref{eq: PSLQ ws lemma total to prime at most original} holds and since $\omega_y(t^{c'},\mathbf{e}^{\succ'})>l(y)$, we must have
\begin{equation*}
    \sum_{x \mathcal{A}(t^{c'})} \Big( l(x) - \omega_x(t^{c'},\mathbf{e}^{\succ'}) \Big)_{+} < \sum_{x \in \mathcal{A}(t^{c'})} \Big( l(x) - \omega_x(t^{c'},\mathbf{e}^{\succ}) \Big)_{+}
\end{equation*}
 a contradiction.
\end{proof}

\begin{Lemma}\label{lemma: PSLQ ws shifted after tight when alone cannot benefit}
Fix $p=\phi_i(P)$ and suppose $t^c < t_p < 1$. Moreover, suppose $i$ is the only student whose top choice is $p$ under $\succ$. If the assignment to $i$ under $\succ$ stochastically dominates that under $\succ'$, there is no $q \in P \setminus \{p\}$ such that $q \succ_i' p$.
\end{Lemma}

\begin{proof}
Suppose, for the sake of reaching a contradiction, that there is such a project $q$. 

We claim this implies $t^{c'} \leq t^c$. Indeed, suppose $t^{c'} > t^c$. Hence,

\begin{equation*}
    n(1-t^c) = \sum_{x \in \mathcal{A}(t^c)} \Big( l(x) -  \omega_x(t^c,\mathbf{e}^{\succ}) \Big)_{+} > \sum_{x \in P} \Big( l(x) -  \omega_x(t^c,\mathbf{e}^{\succ'}) \Big)_{+}
\end{equation*}

\begin{equation*}
    = \sum_{x \in \mathcal{A}(t^c)} \Big( l(x) -  \omega_x(t^c,\mathbf{e}^{\succ'}) \Big)_{+} + \sum_{x \in P \setminus \mathcal{A}(t^c)} \Big( l(x) -  \omega_x(t^c,\mathbf{e}^{\succ'}) \Big)_{+}
\end{equation*}

Since $\sum_{x \in P \setminus \mathcal{A}(t^c)} \Big( l(x) -  \omega_x(t^c,\mathbf{e}^{\succ'}) \Big)_{+} \geq 0$, it follows that

\begin{equation*}
    \sum_{x \in \mathcal{A}(t^c)} \Big( l(x) -  \omega_x(t^c,\mathbf{e}^{\succ}) \Big)_{+} > \sum_{x \in \mathcal{A}(t^c)} \Big( l(x) -  \omega_x(t^c,\mathbf{e}^{\succ'}) \Big)_{+}
\end{equation*}

By definition of $\mathcal{A}(t^c)$, $\omega_x(t^c,\mathbf{e}^{\succ}) < l(x), \forall x \in \mathcal{A}(t^c)$. Hence, 

\begin{equation*}
    \sum_{x \in \mathcal{A}(t^c)} \omega_x(t^c,\mathbf{e}^{\succ'}) > \sum_{x \in \mathcal{A}(t^c)} \omega_x(t^c,\mathbf{e}^{\succ}) 
\end{equation*}

By Lemma~\ref{lemma: PSLQ ws if i 0 then under misreport will not ever decrease}, $\omega_q(t^c,\mathbf{e}^{\succ'}) \geq \omega_q(t^c,\mathbf{e}^{\succ}), \forall q \notin \mathcal{A}(t^c)$. But then 

\begin{equation*}
    n \cdot t^c = \sum_{x \in P} \omega_x(t^c,\mathbf{e}^{\succ'}) = \sum_{x \in P \setminus \mathcal{A}(t^c)} \omega_x(t^c,\mathbf{e}^{\succ'}) + \sum_{x \in \mathcal{A}(t^c)} \omega_x(t^c,\mathbf{e}^{\succ'})
\end{equation*}

\begin{equation*}
     > \sum_{x \in P \setminus \mathcal{A}(t^c)} \omega_x(t^c,\mathbf{e}^{\succ}) + \sum_{x \in \mathcal{A}(t^c)} \omega_x(t^c,\mathbf{e}^{\succ}) = \sum_{x \in P} \omega_x(t^c,\mathbf{e}^{\succ}) = n \cdot t^c
\end{equation*}

a contradiction. To sum up, we have just shown that $t^{c'} \leq t^c$.

Hence, since the preferences of the rest of the students do not change between the profiles, students who shift to $p$ must shift there under $\succ'$ at least as soon as under $\succ$ as shown in Lemma~\ref{lemma: PSLQ ws shifted come earlier}. It follows that

\begin{equation*}
    \sum_{j \in N \setminus \{i\}} \omega_{jp}(t_p,\mathbf{e}^{\succ'}) \geq \sum_{j \in N \setminus \{i\}} \omega_{jp}(t_p,\mathbf{e}^{\succ})
\end{equation*}

Since we assume $R_i' \ sd(\succ_i) \ R_i$, it must follow that $t_p' > t_p$. Together with the fact that $\omega_{ip}(t,\mathbf{e})$ is non-decreasing in $t$, we get

\begin{equation*}
    \sum_{j \in N \setminus \{i\}} \omega_{jp}(1,\mathbf{e}^{\succ'}) = \sum_{j \in N \setminus \{i\}} \omega_{jp}(t_p',\mathbf{e}^{\succ'}) > \sum_{j \in N \setminus \{i\}} \omega_{jp}(t_p,\mathbf{e}^{\succ}) = \sum_{j \in N \setminus \{i\}} \omega_{jp}(1,\mathbf{e}^{\succ})
\end{equation*}

By Lemma~\ref{lemma: PSLQ ws cosntraints tight before then lower anyway}, $\omega_p(1,\mathbf{e}^{\succ}) = \omega_p(1,\mathbf{e}^{\succ'}) = l(p)$. This implies that the total share of $l(p)$ eaten by $i$ must decrease, i.e. $r_{ip}'<r_{ip}$, a contradiction to $R_i' \ sd(\succ_i) \ R_i$.

\end{proof}

\begin{Lemma}\label{lemma: PSLQ ws other projects cannot have less}
Fix $p=\phi_i(P)$ and suppose $t^c = t_p$. Let $q \in P$ such that $r_{iq} = 0$. Suppose the assignment to $i$ under $\succ'$ stochastically dominates the assignment under $\succ$. Then, $\omega_q(1,\mathbf{e}^{\succ}) \leq \omega_q(1,\mathbf{e}^{\succ'})$.
\end{Lemma}

\begin{proof}
By construction and since $R_i' \ sd(\succ_i) \ R_i$, we must have $t_{p}' \geq t_{p}$. Hence, $t^{c'} \geq t_{p}' \geq t_{p} = t^c$. Indeed, this is straightforward if $i$ is the only student who ever eats $p$ under $\succ$. If student $i$ is not the only one ever eating $p$, that is $\exists j \in N \setminus \{i\}, t \in [0,t_{p})$ such that $j \in N_{p}(t,\mathbf{e}^{\succ})$, it follows by Lemma~\ref{lemma: PSLQ ws others continue eating} that $N_p(t,\mathbf{e}^{\succ}) \setminus \{i\} \subseteq N_p(t,\mathbf{e}^{\succ}) \setminus \{i\}, \forall t \in [0,t_p)$. Since $R_i' \ sd(\succ_i) \ R_i$, and $\omega_p$ is non-decreasing in $t$, it follows that $\omega_p(t_p',\mathbf{e}^{\succ'}) \geq \omega_p(t_p,\mathbf{e}^{\succ}) \geq l(p)$. Suppose $t^{c'} < t^c$; this implies $\omega_p(t_p',\mathbf{e}^{\succ'}) = \omega_p(t_p,\mathbf{e}^{\succ}) = l(p)$. Since the preferences of the rest of the agents do not change between the profiles, for $R_i' \ sd(\succ_i) \ R_i$, we require $N_p(t,\mathbf{e}^{\succ}) \subseteq N_p(t,\mathbf{e}^{\succ'}), \forall t \in [0,t_p)$. By Lemma~\ref{lemma: PSLQ ws helpful can't have more}, it follows that $t_p = t_p'$ and hence the eating algorithms must coincide on $[0,t_p)$, which contradicts $t^{c'} < t^c$.

The result now follows by Lemma~\ref{lemma: PSLQ ws if i 0 then under misreport will not ever decrease}.

\end{proof}

We turn to the proof of the theorem.

\subsubsection*{PSLQ is weakly strategy-proof}

\begin{proof}
Fix student $i \in N$ and suppose $i$ misreports to $\succ_i'$. For notational ease, label $P$ such that $p_1 \succ_i p_2 \succ_i \dots \succ_i p_k$. We write $R=PSLQ(\succ)$ and $R'=PSLQ(\succ')$. Suppose $R_i' \ sd(\succ_i) \ R_i$. Our goal is to show $R_i=R_i'$. 

\begin{comment}
We use a proof by contradiction. In particular, we assume $R_i \neq R_i'$. Then, by Lemma~\ref{lemma: proof PSLQ ws 1}, there must exist some $1 \leq s \leq k$ such that $r_{ip_s}'>r_{ip_s}$ and $r_{ip_l} = r_{ip_l}', \forall l < s$. We show that no such $s$ can exist.
\end{comment}

If $r_{ip_1} = 1$, we are done. Suppose $r_{ip_1}<1$. Note that under $\succ$, student $i$ is eating $p_1$ throughout $[0,t_{p_1})$. Thus, $r_{ip_1} = t_{p_1}$. On the other hand, student $i$ is eating $p_1$ on a subset of $[0,t_{p_1}')$. Hence, $r_{ip_1}' \geq r_{ip_1}$ leads to $t_{p_1} \leq t_{p_1}'$.

Suppose $t_{p_1}<t^c$, so we must have $\omega_{p_1}(1,\mathbf{e}^{\succ}) = u(p_1)$, i.e. when students cease to eat from $p_1$, it is because they hit the upper quota. Note that this implies that student $i$ is not the only one who ever eats from $p_1$, as $u(p_1) \geq 1$ and $r_{ip_1}<1$ by assumption. By Lemma~\ref{lemma: PSLQ ws others continue eating}, it follows that $N_{p_1}(t,\mathbf{e}^{\succ}) \setminus \{i\} \subseteq N_{p_1}(t,\mathbf{e}^{\succ}) \setminus \{i\}, \forall t \in [0,t_{p_1})$. Since $R_i' \ sd(\succ_i) \ R_i$, and $\omega_{p_1}$ is non-decreasing in $t$, it follows that $\omega_{p_1}(t_{p_1}',\mathbf{e}^{\succ'}) \geq \omega_{p_1}(t_{p_1},\mathbf{e}^{\succ}) = u(p_1) \Rightarrow \omega_{p_1}(t_{p_1}',\mathbf{e}^{\succ'}) = u(p_1)$. For $R_i' \ sd(\succ_i) \ R_i$, we therefore require $N_{p_1}(t,\mathbf{e}^{\succ}) \subseteq N_{p_1}(t,\mathbf{e}^{\succ'}), \forall t \in [0,t_{p_1})$. By Lemma~\ref{lemma: PSLQ ws helpful can't have more}, it follows that $t_{p_1} = t_{p_1}'$ and hence $N_{p_1}(t,\mathbf{e}^{\succ'}) = N_{p_1}(t,\mathbf{e}^{\succ})$ on the interval $[0,t_{p_1})$ to ensure feasibility. Therefore, $r_{ip_1} = r_{ip_1}'$ and $PSLQ(\succ)$ and $PSLQ(\succ')$ must coincide on $[0,t_{p_1})$.

This argument can be repeated as we iterate through the sequence $\{p_1,\dots,p_k\}$, as long as students are not shifted from the particular project. We now turn our attention to this case. Let $1\leq l \leq k$ denote the index of the project such that $\omega_{p_j}(1,\mathbf{e}^{\succ}) = u(p_j), \forall j < l$ and $p_l$ is shifted from. By the argument from the preceding paragraph, we must have that the two eating algorithms, for $\succ$ and $\succ'$, coincide on $[0,t_{p_{l-1}})$. We can interpret this as the algorithms effectively restart at $t_{p_{l-1}}$ and proceed on the following sub-problem: the new epoch has length $1-t_{p_{l-1}}$, fully eaten projects are removed and the upper and lower quota of any other project $q$ is lowered by $\omega_q(t_{p_{l-1}},\mathbf{e}^{\succ})$. Since the strategic considerations are the same in the original problem and the sub-problem, we may without loss of generality assume that students are shifted from $p_1$ at $t_{p_1}$.

We distinguish two cases: $t^c < t_{p_1}$ and $t^c = t_{p_1}$. Suppose the former holds. Then by Lemma~\ref{lemma: PSLQ ws cosntraints tight before then lower anyway}, we have $\omega_{p_1}(1,\mathbf{e}^{\succ'})=\omega_{p_1}(1,\mathbf{e}^{\succ}) = l(p_1)$. If there is another student $j \in N \setminus \{i\}$ such that $p_1$ is $j$'s top choice, it follows by a similar argument used two paragraphs above, using Lemmas~\ref{lemma: PSLQ ws helpful can't have more} and ~\ref{lemma: PSLQ ws others continue eating}, that $t_{p_1} = t_{p_1}'$ and hence $N_{p_1}(t,\mathbf{e}^{\succ'}) = N_{p_1}(t,\mathbf{e}^{\succ})$ on the interval $[0,t_{p_1})$. Therefore, $r_{ip_1} = r_{ip_1}'$. Otherwise, suppose $i$ is the only student whose top choice is $p_1$ under $\succ$. By Lemma~\ref{lemma: PSLQ ws shifted after tight when alone cannot benefit}, it follows that $i$ eats $p_1$ throughout $[0,t_{p_1})$ under $\succ'$ and $r_{ip_1}'=r_{ip_1}$. In either case, the eating algorithms coincide on $[0,t_{p_1})$. 

Now suppose $t^c = t_{p_1}$. We claim that the total amount eaten from $p_1$ under $\succ'$ is at most the amount under $\succ$, i.e. $\omega_{p_1}(1,\mathbf{e}^{\succ'}) \leq \omega_{p_1}(1,\mathbf{e}^{\succ})$. Indeed, we have

\begin{equation}\label{eq: PSLQ ws other projects no less}
    \omega_q(1,\mathbf{e}^{\succ}) \leq \omega_q(1,\mathbf{e}^{\succ'}), q \in P \setminus \{p_1\}
\end{equation}

Let us write $\mathcal{A}$ shorthand for $\mathcal{A}(t_{p_1})$, the set of active projects at $t_{p_1}$.
It is clear that $l(q) = \omega_q(1,\mathbf{e}^{\succ}) \leq \omega_q(1,\mathbf{e}^{\succ'}), \forall q \in \mathcal{A}$ by feasibility. Consider any $q \notin \mathcal{A}$ such that $q \neq p_1$. Then by Lemma~\ref{lemma: PSLQ ws shifted structure vector}, we have $r_{iq}=0$. Applying Lemma~\ref{lemma: PSLQ ws other projects cannot have less}, we obtain the inequality.

Next, we distinguish two cases. First, suppose student $i$ is not the only one ever eating $p_1$, that is $\exists j \in N \setminus \{i\}, t \in [0,t_{p_1})$ such that $j \in N_{p_1}(t,\mathbf{e}^{\succ})$; let us denote the set of such students by $J$. By Lemma~\ref{lemma: PSLQ ws others continue eating} it follows that $i \neq j \in N_{p_1}(t,\mathbf{e}^{\succ}) \Rightarrow j \in N_{p_1}(t,\mathbf{e}^{\succ'}), \forall t \in [0,t_{p_1})$. Hence, since $t_{p_1} \leq t_{p_1}'$ by assumption, we also have $r_{jp_1}' \geq r_{jp_1}, \forall j \in J$, with the inequality being strict if $t_{p_1} < t_{p_1}'$. So suppose, for the sake of reaching a contradiction, that this is indeed the case. But since the total assignment of $p_1$ is under $\succ'$ at most the total assignment under $\succ$, i.e. $\omega_{p_1}(1,\mathbf{e}^{\succ'}) \leq \omega_{p_1}(1,\mathbf{e}^{\succ})$, this implies that $r_{ip_1}' < r_{ip_1}$, a contradiction. We therefore conclude that $t_{p_1} = t_{p_1}'$ and hence $N_{p_1}(t,\mathbf{e}^{\succ'}) = N_{p_1}(t,\mathbf{e}^{\succ})$ on the interval $[0,t_{p_1})$. Further, it follows that $r_{ip_1} = r_{ip_1}'$ and $PSLQ(\succ)$ and $PSLQ(\succ')$ must coincide on $[0,t_{p_1})$. 

Since the algorithms coincide on $[0,t_{p_1})$ and since $t^c \leq t_{p_1}$, we must have $t^{c'}=t^c$, and so throughout the rest of the execution of the algorithm

\begin{equation*}
    n(1-t) = \sum_{p \in P} \Big( l(p) - \omega_p(t,\mathbf{e}^{\succ}) \Big)_+, \forall t \in [t_{p_1},1]
\end{equation*}

In other words, the algorithm on this sub-problem behaves exactly like the PS mechanism. As \cite{bogomolnaia2001new} proved, this mechanism is weakly strategy-proof, so we may conclude that $R_i = R_i'$.

What remains is the case when $i$ is the only student who ever eats from $p_1$ under $\succ$. Note that since $r_{ip_1} = \omega_{p_1}(1,\mathbf{e}^{\succ}) < 1$, we must have $l(p_1) = 0$. Moreover, $r_{ip_1}' \geq r_{ip_1}$ implies $\omega_{p_1}(1,\mathbf{e}^{\succ'}) \geq \omega_{p_1}(1,\mathbf{e}^{\succ})$. However, recall we already showed $\omega_{p_1}(1,\mathbf{e}^{\succ'}) \leq \omega_{p_1}(1,\mathbf{e}^{\succ})$, so we conclude $\omega_{p_1}(1,\mathbf{e}^{\succ'}) = \omega_{p_1}(1,\mathbf{e}^{\succ})$ and $r_{ip_1}' = r_{ip_1}$. Moreover, note that if the top choice of $i$ under $\succ_i$ and $\succ_i'$ is $p_1$, the eating algorithms again coincide on $[0,t_{p_1})$ and the result follows by the same argument as in the paragraph above. In what follows, suppose this is not the case.

The preceding paragraph also implies that there cannot exist another project $q$ such that $l(q) \leq \omega_q(t_{p_1},\mathbf{e}^{\succ}) < u(q)$ (this then implies $q \notin \mathcal{A}$). This follows from Equation~\ref{eq: PSLQ ws other projects no less} and the fact that no student would be required to shift under $\succ'$ before $t_{p_1}' > t_{p_1}$, where the inequality is strict because the top choice under $\succ_i'$ is not $p_1$ by assumption. If such a $q$ existed, since $\omega_q$ is non-decreasing in $t$ and $t^{c'} \geq t_{p_1}' > t_{p_1} = t^c \Rightarrow t_q' > t_q$, we would have
\begin{equation*}
\omega_q(1,\mathbf{e}^{\succ'}) = \omega_q(t_q',\mathbf{e}^{\succ'}) > \omega_q(t_q,\mathbf{e}^{\succ'})  \geq \omega_q(t_q,\mathbf{e}^{\succ}) = \omega_q(1,\mathbf{e}^{\succ})    
\end{equation*}
But then, since we assume both assignments are feasible,
\begin{equation*}
    n = \sum_{p \in P} \omega_p(1,\mathbf{e}^{\succ}) < \sum_{p \in P} \omega_p(1,\mathbf{e}^{\succ'}) = n
\end{equation*}
a contradiction.

To sum up, we have shown that if $i$ is the only student who ever eats from $p_1$ under $\succ$, the total allocation $\omega_q(1,\mathbf{e}^{\succ})$ of any project $q \in P\setminus \{p_1\}$ is either $l(q)$, if students are shifted to $q$ at some point during the execution of the algorithm, or $u(q)$, if the project is fully eaten.

This leads us to the following conclusion. We have shown that under no other configuration of parameters than in the preceding paragraph, we can have $R_i \neq R_i'$. However, since $n \in \mathbb{N}$, $0 < \omega_{p_1}(1,\mathbf{e}^{\succ}) < 1$, and the lower and upper quotas are integers, this configuration is impossible.

Note that our assumption above where we assumed that, without loss of generality, $p_1$ is the first project $p$  $i$ is eating such that students shift from $p$ is indeed without loss of generality and does not invalidate our finding here. If we took another project $p_l$ with $l>1$ instead, our analysis implies $\omega_{p_s}(1,\mathbf{e}^{\succ})$ is either $l(p_s)$ or $u(p_s)$, for all $s=l+1,\dots,k$ in the sub-problem. These are not necessarily integers anymore, but recall we subtracted the allocations of such projects obtained until $t_{p_{l-1}}$ from their quotas. Therefore, their total allocation in the original problem is either $l(p_s)$ or $u(p_s)$; in particular an integer.

\end{proof}

\subsection{Strategic manipulation under non-integer quotas}\label{sec: Strategic manipulation under non-integer quotas}

\begin{Ex}\label{ex: PSLQ not ws}
Let $N=[2]$, $P=\{a,b\}$ and suppose the projects and students have the following quotas and preferences, respectively.
\[
\begin{tabular}{ c|ccc }
& a & b \\
  \hline
$u(\cdot)$ & $\infty$ & 2/3 & 2/3\\
$l(\cdot)$ & 0 & 2/3 & 2/3 \\
\end{tabular}
\qquad
\succ \ = 
    \begin{tabular}{cc}
    \hline
    1 & 2 \\
    \hline
    a & b \\
    b & c \\
    c & a \\
    \hline
    \end{tabular}
\]

The resulting random assignment matrix is $R$. Now suppose 1 misreports to $\succ_1': b \succ_1' a \succ_1' c$. The resulting random assignment matrix is $R'$ and we obtain $R_1' \ sd(\succ_1) \ R_1$ while $R_1 \neq R_1'$.

\[
R=
  \begin{blockarray}{*{3}{c} l}
    \begin{block}{*{3}{>{\footnotesize}c<{}} l}
      a & b & c \\
    \end{block}
    \begin{block}{[*{3}{c}]>{\footnotesize}l<{}}
      2/3 & 0 & 1/3 & 1\\
      0 & 2/3 & 1/3 & 2\\
    \end{block}
  \end{blockarray}
\qquad
R'=
  \begin{blockarray}{*{3}{c} l}
    \begin{block}{*{3}{>{\footnotesize}c<{}} l}
      a & b & c \\
    \end{block}
    \begin{block}{[*{3}{c}]>{\footnotesize}l<{}}
      2/3 &  1/3 & 0 & 1\\
      0 & 1/3 & 2/3 & 2\\
    \end{block}
  \end{blockarray}
\]

\end{Ex}

\section{Proposition~\ref{impossibility result}}\label{proof impossibility result}

\begin{proof}
Suppose there is a random assignment mechanism which is ordinally efficient, envy-free and weakly-strategyproof. Let us denote this mechanism by $R=R(\succ)$. We aim to obtain a contradiction by considering the setup in Example~\ref{ex: PSLQ not ws} and the restriction the desired properties impose on the set of random assignment matrices at different preference profiles. 

Under $\succ$, ordinal efficiency implies that $r_{2a}=0$ and $r_{1a}=2/3$. Indeed, suppose $r_{1b}+r_{1c} > 1/3$ and that $r_{1b}>0$ (the case $r_{1c}>0$ is similar). Then students $1$ and $2$ and projects $a$ and $b$ form a wasteful chain, which implies $R$ is not ordinally efficient by Lemma~\ref{lemma: characterisation ordinal efficiency}, a contradiction. Therefore, $r_{1a}=2/3$, $r_{2a}=0$. Let us parameterise $r_{1b}=t$ and $r_{1c}=1/3-t$ for $t \in [0,1/3]$. Hence, quotas on $b$ and $c$ and feasibility of $R$ imply $r_{2b}=2/3-t$ and $r_{2c}=1/3+t$. Let us denote the random assignment matrix for a particular $t \in [0,1/3]$ by $R^t$, that is

\[
R^t=
  \begin{blockarray}{*{3}{c} l}
    \begin{block}{*{3}{>{\footnotesize}c<{}} l}
      a & b & c \\
    \end{block}
    \begin{block}{[*{3}{c}]>{\footnotesize}l<{}}
      2/3 & t & 1-t & 1\\
      0 & 2/3-t & 1/3+t & 2\\
    \end{block}
  \end{blockarray}
\]

It is straightforward to check that for any $t \in [0,1/3]$, $R_t$ is ordinally efficient and envy-free.

Suppose student 1 misreports to $\succ_1': b \succ_1' a \succ_1' c$, while $\succ_2$ remains unchanged; let us denote this profile by $\succ'=(\succ_1',\succ_2)$ and $R'=R(\succ')$. By envy-freeness, $r_{1b}'=r_{2b}'=1/3$ and by ordinal efficiency, $r_{1a}=2/3$ (using a similar argument to the paragraph above). It follows the following is the unique ordinally efficient and envy-free random assignment for the preference profile $\succ'$

\[
R'=
  \begin{blockarray}{*{3}{c} l}
    \begin{block}{*{3}{>{\footnotesize}c<{}} l}
      a & b & c \\
    \end{block}
    \begin{block}{[*{3}{c}]>{\footnotesize}l<{}}
      2/3 & 1/3 & 0 & 1\\
      0 & 1/3 & 2/3 & 2\\
    \end{block}
  \end{blockarray}
\]

Now suppose student 2 misreports to $\succ_2': b \succ_2' a \succ_2' c$, while $\succ_1$ remains unchanged; let us denote this profile by $\succ''=(\succ_1,\succ_2')$ and $R''=R(\succ'')$. Envy-freeness implies $r_{1a}''+r_{1b}''=r_{2a}''+r_{2b}''$, which in turn implies $r_{1c}''=r_{2c}''=1/3$. Ordinally efficiency implies $r_{1a}''=r_{2b}''=2/3$. Therefore, the following is the unique ordinally efficient and envy-free random assignment for $\succ''$

\[
R''=
  \begin{blockarray}{*{3}{c} l}
    \begin{block}{*{3}{>{\footnotesize}c<{}} l}
      a & b & c \\
    \end{block}
    \begin{block}{[*{3}{c}]>{\footnotesize}l<{}}
      2/3 & 0 & 1/3 & 1\\
      0 & 2/3 & 1/3 & 2\\
    \end{block}
  \end{blockarray}
\]

Observe that for any $t \in [0,1/3]$, $R_1' \ sd(\succ_1) \ R_1^t$ and $R_2'' \ sd(\succ_2) \ R_2^t$. Since we require $R$ to be weakly strategy-proof, we need to choose a suitable $t^* \in [0,1/3]$ so that $R_1'=R_1^{t^*}=R_1$ and $R_2''=R_2^{t^*}=R_2$. But

\begin{equation*}
    R_1' \ sd(\succ_1) \ R_1^{t^*} \text{ and } R_1'=R_1^{t^*} \Rightarrow t^* = 1/3
\end{equation*}

and 

\begin{equation*}
    R_2'' \ sd(\succ_2) \ R_2^{t^*} \text{ and } R_2''=R_2^{t^*} \Rightarrow t^* = 0
\end{equation*}

a contradiction. Therefore, we have found a market for which no random assignment mechanism is ordinally efficient, envy-free and weakly strategy-proof. The result follows.

\end{proof}

\end{document}